\documentclass[12pt]{amsart}

\usepackage{graphicx}
\usepackage{amsmath,amsthm}
\makeatletter
\def\amsbb{\use@mathgroup \M@U \symAMSb}
\makeatother
\usepackage{bbold}
\usepackage[all]{xy}
\usepackage{amssymb}
\usepackage{caption}
\usepackage[latin1]{inputenc}
\usepackage{enumerate}
\usepackage{mathrsfs}  
\usepackage{dsfont}
\usepackage[bbgreekl]{mathbbol}
\usepackage{geometry}
\usepackage{datetime}

\theoremstyle{plain}

\newtheorem{thm}{Theorem}[section]
\newtheorem{cor}{Corollary}[section]
\newtheorem{lemma}{Lemma}[section]

\newtheorem{defi}{Definition}[section]

\newtheorem{obs}{Remark}[section]

\newtheorem*{obsparticulares}{Remark on $\sigma^r(P)$ for $P$ of order $r$}

\numberwithin{equation}{section}

\newcommand{\DP}[2]{\frac{\partial\, #1}{\partial\, #2}}

\def\C{\mathcal {C}}
\def\Com{\mathbb{C}}

\def\A{{\mathcal{C}}^\infty(M)}

\def\AC{{\mathcal{C}}^\infty(T^*M)}

\def\CE{{\mathcal{C}}^\infty(E)}
\def\t{\mathbb{t}}
\def\M{\mathds{M}}

\def\H{\mathds{H}}
\def\T{\mathds{T}}

\def\Y{\mathds{Y}}
\def\Z{\mathds{Z}}
\def\DD{\mathds{D}}
\def\SS{\mathds{S}}

\def\U{\mathcal {U}}

\def\R{\amsbb{R}}
\def\m{\mathfrak m}

\def\Ddelta{\mathbb{\Delta}}

\def\be{\begin{equation}
}

\def\ee{\end{equation}
}

\begin{document}

\title[Time in classical and quantum mechanics]
 {Time in classical and quantum mechanics}

\author{ J.  Mu\~{n}oz-D{\'\i}az and R. J.  Alonso-Blanco}

\address{Departamento de Matem\'{a}ticas, Universidad de Salamanca, Plaza de la Merced 1-4, E-37008 Salamanca,  Spain.}
\email{clint@usal.es, ricardo@usal.es}

\begin{abstract}
In this article we study the nature of time in Mechanics. The fundamental principle, according to which a mechanical system evolves governed by a second order differential equation, implies the existence of an \emph{absolute time-duration} in the sense of Newton. There is a second notion of time for conservative systems which makes the Hamiltonian action evolve at a constant rate.  In Quantum Mechanics the absolute time loses its sense as it does the notion of trajectory.
  Then, we propose two different ways to reach the time dependent Schrödinger equation. One way consists of considering a ``time constraint'' on a free system. The other way is based on the point of view of Hertz, by considering the system as a projection of a free system. In the later manner, the ``time'' appearing in the Schrödinger equation is a linear combination of the time-duration with the ``time'' quotient of the action by the energy on each solution of the Hamilton-Jacobi equation. Both of them are based on a rule of quantization that we explain in Section \ref{tres}.
\end{abstract}
\bigskip

\maketitle

\setcounter{tocdepth}{1}


\tableofcontents

\section{Introduction}\label{intro}

Let us recall the Scholium to the Definitions in the Book I of the \emph{Principia} of Newton:
{\it
``Absolute, true, and mathematical time, of itself, and from its own nature, flows equably without relation to anything external, and by another name is called duration.''}

Classical Mechanics has been developed since Newton without the need for an analysis of the nature of time. H. Hertz,  who wrote his \emph{Die Prinzipien der Mechanik} \cite{Hertz} with a striving for rigor, avoiding ambiguities in the notion of force, dispatches the question of time in a line at the beginning of Chapter I of his book: ``The time of the first book is time of our inner intuition...'' In the excellent text of Arnold \cite{Arnold} we read: ``Space and time. Our space is three-dimensional and Euclidean, and time is one-dimensional''. In the \emph{Mechanik} of E. Mach \cite{Mach}, Ch.II,  \S 6, there is a criticism of the absolute time of Newton: ``It is an idle
metaphysical conception...'' But no contribution is offered that clarifies the meaning of ``time''.

The decisive breakthrough in the formalization of Newtonian mechanics is given in the book of Lagrange \cite{Lagrange}. From then on, and using the language used today, a mechanical system is a differentiable manifold $M$ (the configuration space), equipped with a Riemannian metric $T_2$, which collects the mechanical properties (masses, moments of inertia, ... ) of the system, and whose tangent bundle  $TM$ is the space of position-velocity states of the system. The old law ``Force = mass $\times$ acceleration'' is carried to its general form: the evolution of a mechanical system is the flow of a field $D$ tangent to $TM$, that is a second order differential equation. In the evolution of the system, the ``time'' is the very parameter of the trajectories of the field $D$.

However, there is no function in $M$ that can parameterize all the trajectories of a differential equation of second order: there can not be a  ``time'' function  on $M$. Nor can there be a time function on $ TM $ that can parameterize all the trajectories of all the second order differential equations.

``Time'' is not, nor can it be a function. It is a class of differential 1-forms on $TM$ that behaves in a functorial way. In physical terms: if a trajectory of a system is projected into another smaller system, ``ignoring degrees of freedom'', the duration does not change. That is exactly what the ``absolute'' character of Newtonian time consists of.

``Time'' is not a pre-existing reality before mechanical systems. It is a mathematical object that is into the very structure of Classical Mechanics.

This ``time'' is ``duration'',  as already noted by Newton. The issue of ``simultaneity'' is a different one.

In fact, the time-duration  remains in Relativity. Relativistic mechanical systems are a particular type within Classical Mechanics: those in which the parameter of the trajectories is the length (= proper time) \cite{RelatividadMunoz,RM}. ``Absolute time'' still exists in Relativity.

The sharp problem about the nature of time is presented in Quantum Mechanics. When the notion of ``trajectory'' stops making sense, what can be  time?

The Hamiltonian of a quantum-mechanical system can be considered as ($i$ times) the generator of a uniparametric group of unitary automorphisms and call the parameter ``time'' \cite{Yosida}. But, in that way, ``time is just treated as an \emph{external parameter} in standard quantum mechanics, rather than as a dynamical variable'' \cite{Penrose}, p. 524. The problem is in the relation that this parameter has with the time of the Classical Mechanics.

In the present article, we establish a general rule of quantization for contravariant tensor fields in any variety, which provides the quantization of the magnitudes of classical mechanical systems. In a precise sense, we will see that equations ``of Schrödinger''  are the only ones that are related to  second-order differential equations of the classical mechanical systems (and that, moreover, they must be conservative).

Once given the quantization rule, we propose two ways of arriving at the Schrödinger equation with time. Each of these modes is deduced from a classical mechanical method to obtain a conservative system from a geodesic one.

The first method consists in imposing a so-called \emph{time constraint} on a geodesic system: the space of states of the classical system is limited to those in which the function imposed as ``time'' evolves uniformly, like Newtonian time. In the corresponding Schrödinger equation, the parameter ``time'' is that imposed in the constraint. There is a direct interpretation of quantum time as classical time.

In the second method, the conservative system is the projection of a geodesic system in greater dimension by means of a \emph{Hertz constraint}. Here two ``times'' appear. One is the Newtonian absolute time, the \emph{time of the particles}, which parameterizes the virtual trajectories of the particles. Another is the \emph{wave time}, which parameterizes the progression of the wave fronts in each solution of the Hamilton-Jacobi equation. In the Schrödinger equation that is obtained in this way, the ``time'' that appears is a linear combination of particle time and wave time. This ``time'' is interpretable in classical terms only within each energy level, not in the superposition of states, in general.

The interpretation of the formula $E=h\,\nu$ is ambiguous if we do not know to what time  the frequency $\nu$ refers to. The case of the classical one particle systems in which the trajectories are always closed under a given energy level (Kepler and harmonic oscillator (Bertrand Theorem \cite{Tisserand})) is discussed in the last point of this article.

\section{Time as duration and its absolute character}\label{time}

A classical \emph{mechanical system} is a finite dimensional smooth manifold $M$, endowed with a riemannian metric $T_2$ (non degenerate of arbitrary signature) and a tangent vector field $D$ on its tangent bundle $TM$, which is a second order differential equation. $M$ is the \emph{configuration space}, $\dim M$ is the number of degrees of freedom, $T_2$ incorporates the dynamical properties (masses, inertial momenta, etc.), and $D$ is the evolution law in the space of position-velocity states, $TM$. Each solution-curve of $D$ is the evolution of the system from an initial state, being the parameter of the curve the one being canonically associated with $D$. This parameter is the \emph{time}. This is the statement of the Mechanics which we owe, essentially, to Lagrange \cite{Lagrange}.

The time-duration is not related to the metric $T_2$. Its structure strictly refers to the position-velocity space, $TM$.

To begin with, let us point out that there is no function in $ M $ that can be used to parameterize all the solution-curves of the second order differential equation $D$. There is no ``time'' function in $ M $. The object that parameterizes all the solution curves of all the second order differential equations in $ M $ is a class of 1-forms in $ TM $, which we call the class of time \cite{MecanicaMunoz}. The class of time has a functorial behavior
with respect to morphisms of manifolds: if $\varphi\colon M\to N$ is a morphism of manifolds, and $\varphi_*\colon TM\to TN$ is the corresponding morphism of tangent bundles, for each open set $\U\subset TN$ and each 1-form $\tau$ in the class of time in $ \U$, then the 1-form $(\varphi_*)^*\tau$ belongs to the class of time in $(\varphi_*)^{-1}(\U)$. This is the precise meaning of the ``absolute'' character of the time-duration, as will become clear later.
\bigskip

Let us go into the details.

Tangent vectors or tangent fields on $TM$ that are tangent to the fibres of the projection $\pi\colon TM\to M$ will be called \emph{vertical}. Those are just the ones which, as derivations of the ring $\C^\infty(TM)$, annihilate the subring $\C^\infty(M)$.

Each fibre $T_aM$ is a vector space, so that it can be identified with its own tangent space at each point: each vector $v_a\in T_aM$ determines another vector $V_{u_a}\in T_{u_a}T_aM$ at each $u_a\in T_aM$; $V_{u_a}$ is the derivative along the vector $v_a$. We will say that $V_{u_a}$ is the \emph{vertical representative} of $v_a$ (at $u_a$) and that $v_a$ is the \emph{geometrical representative} of $V_{u_a}$.

That correspondence assigns to each vertical tangent field $V$ on $TM$ a field $v$ on $TM$ with values in $TM$, that is to say, a section of $TM\times_M TM\to TM$ (projection over the first factor).

The differential 1-forms on $TM$ that (by interior product) annihilate all the vertical vectors will be called \emph{horizontal 1-forms}. For each function $f\in\C^\infty(M)$, $df$ is a horizontal 1-form and, locally, any horizontal 1-form is a linear combination of the forms $df$, with coefficients in $\C^\infty(TM)$.

Each horizontal 1-form $\alpha$ on $TM$ defines a function $\dot\alpha$ on $TM$ by the rule $\dot\alpha(v_a)=\langle\alpha,v_a\rangle$, for each $v_a\in TM$. In particular, for each function $f\in\C^\infty(M)$, the function $\dot{\overline {df}}$ will be denoted, for short, as $\dot f$. For each tangent vector $v_a\in TM$ we have $\dot f(v_a)=v_a(f)$. The map $\dot d\colon\C^\infty(M)\to\C^\infty(TM)$, $f\mapsto\dot df:=\dot f$ is, essentially, the differential.

If $\{x^1,\dots,x^n\}$ are local coordinates on an open set $\U$ of $M$, $\{x^1,\dots,x^n,\dot x^1,\dots,\dot x^n\}$ are local coordinates for
$T\U\subseteq TM$. And, for each $f\in\C^\infty(M)$ we have, in $\U$,
$$\dot df=\frac{\partial f}{\partial x^j}\,\dot x^j=\left(\dot x^j\frac{\partial }{\partial x^j}\right)f.$$
In this way, the local expression for the field $\dot d$ (field on $TM$ with values  in $TM$) is
  $$\dot d=\dot x^j\frac{\partial }{\partial x^j}.$$

  The vertical representative of $\partial/\partial x^j$ at each point $u_a\in TM$ is $(\partial/\partial\dot x^j)_{u_a}$, as it follows directly from the definitions. Therefore, the vertical field that corresponds to the tautological field $\dot d$ is $V:=\dot x^j\partial/\partial \dot x^j$, the infinitesimal generator of the homotheties in fibres.

  A tangent vector $D_{v_a}\in T_{v_a}(TM)$ is an \emph{acceleration} when, for each $f\in\C^\infty(M)$ is $D_{v_a}f=v_af$; that is to say, when $\pi_*(D_{v_a})=v_a$ ($\pi\colon TM\to M$ is the canonical projection). A field $D$, tangent to $TM$, is an \emph{second order differential equation} when the value $D_{v_a}$ at each $v_a\in TM$ is an acceleration. This means that, for each $f\in\C^\infty(M)$, we have $Df=\dot df=\dot f$. In local coordinates,
  $$D=\dot x^j\,\frac{\partial}{\partial x^j}+f^j(x,\dot x)\,\frac{\partial}{\partial\dot x^j},$$
  for certain functions $f^j\in\C^\infty(TM)$.

  Two second order differential equations $D$, $\overline D$ derive in the same way the subring $\C^\infty(M)$ of $\C^\infty(TM)$. For this reason, $D-\overline D$ is a vertical field. Second order differential equations are the sections of an affine bundle over $TM$ (the \emph{bundle of accelerations}), modeled over the vector bundle of the vertical tangent fields. When a connection on the tangent bundle $TM\to M$ is given (for example, the Levi-Civita connection of a metric given on $M$), the geodesic field $D_G$ is a second order differential equation, which provides an origin for the affine bundle of accelerations. For each second order differential equation $D$, the difference $D-D_G$ is a vertical field. In the language of Physics, vertical fields are the \emph{forces}. Thus, the datum of a connection puts in correspondence the field of accelerations $D$ with the field of forces $D-D_G$ (see the details in Section \ref{classical}).

  A 1-form $\alpha$ on $TM$ is said a \emph{contact form} when annihilates, by interior product,  all the second order differential equations. Such an $\alpha$ also annihilates the difference of any couple of second order differential equations. This is to say, $\alpha$ annihilates every vertical tangent field and, then, is a horizontal 1-form. On the other hand, for each 1-form horizontal $\alpha$ and acceleration $D_{v_a}$, we have
    $$\langle\alpha,D_{v_a}\rangle=\langle\alpha,v_a\rangle=\dot\alpha(v_a),$$
 so that $\alpha$ will be a contact form if and only if $\alpha$ is horizontal and, in addition, $\dot\alpha=\langle\alpha,\dot d\rangle=0$.

 The set of contact 1-forms is a Pfaff system on $TM$, the \emph{contact system} $\Omega$, of rank $n-1$ (if $n=\dim M$), generated out of the 0-section of $TM$, by the 1-forms
 $$\dot x^j\,dx^i-\dot x^i\, dx^j,\quad (i,j=1,\dots,n),$$
 on each coordinated open set.

  For each $v_a\in TM$, the tangent vectors on $TM$ at $v_a$ annihilated by $\Omega$ are the multiples of accelerations at $v_a$, along with the vertical vectors at $v_a$. A curve $\Gamma$ in $TM$ is a solution of the contact system if it is tangent at each point to an acceleration or a (non trivial) vertical vector.

 We will say that a horizontal 1-form $\alpha$ belongs to the \emph{class of time} on an open set $\U$ of $TM$ when $\dot\alpha=1$ on $\U$. If $\alpha$, $\beta$ belong to the class of time on $\U$, then $\alpha-\beta\in\Omega$ on $\U$: two 1-forms in the class of time are congruent modulo the contact system. For each function $f\in\C^\infty(M)$, on the open set of $TM$ where $\dot f\ne 0$, the form $df/\dot f$ belongs to the class of time. It is derived that each point $v_a$, out of the 0-section, has a neighborhood in which there is a form in the class of time: the class of time is defined all along $TM$ except the 0-section. In fact, an argument with partitions of the unity shows that there is a global form $\tau$ in the class of time on the complementary open set of the  0-section in $TM$.

 For each curve $\Gamma$ in $TM$, solution of the contact system, and which does not intersect the 0-section, we will call \emph{duration} of $\Gamma$ the integral $\int_\Gamma\tau$, where $\tau$ is an 1-form in the class of time.

 The ``absolute'' character of the duration is a consequence of the functoriality of the notion of acceleration:
 \begin{lemma}
 If $\varphi\colon M\to N$ is a morphism of smooth manifolds, and $\varphi_*\colon TM\to TN$ is the corresponding morphism between their tangent bundles, the tangent map $\varphi_{**}\colon T(TM)\to T(TN)$ sends accelerations in $TM$ to accelerations in $TN$.
 \end{lemma}
 \begin{proof}
 The commutative diagram
 $$
 \xymatrix{ TM\ar[r]^-{\varphi_*}\ar[d]_-\pi & TN\ar[d]^-{\pi}
     & {}\ar@{}[d]^-{\displaystyle{\text{gives}}} & & T_{v_a}(TM)\ar[r]^-{\varphi_{**}}\ar[d]_-{\pi_*} & T_{\varphi_*(v_a)}(TN)\ar[d]^-{\pi_*}\\
 M\ar[r]_-\varphi & N & & & T_aM\ar[r]_-{\varphi_*} & T_{\varphi(a)}N}
 $$

 Now, if $D_{v_a}$ is an acceleration at $v_a\in TM$ (so that $\pi_*D_{v_a}=v_a$), then
 $$\pi_*(\varphi_{**}(D_{v_a}))=\varphi_*\circ \pi_*\,(D_{v_a})=\varphi_*(v_a),$$
 which shows that $\varphi_{**}(D_{v_a})$ is an acceleration at $\varphi_*(v_a)$.
  \end{proof}

\begin{cor}
$(\varphi_*)^*$ applies the contact system of $N$, $\Omega_N$, into the contact system of $M$, $\Omega_M$. $\varphi_*$ applies curves $\Gamma$ solution of $\Omega_M$ into curves $\varphi_*(\Gamma)$ solution of $\Omega_N$.
\end{cor}

\begin{cor}
For each 1-form $\tau$ in the class of time on an open set $\U$ of $TN$, the form $(\varphi_*)^*(\tau)$ is in the class of time on the open set $\varphi_*^{-1}(\U)$ of $TM$.
\end{cor}

\begin{cor}[``Absolute'' character of the duration]\label{absolutetime}
Let $\Gamma$ be a curve solution of $\Omega_M$, whose image $\varphi_*(\Gamma)$ does not intersect the 0-section of $TN$. The duration of $\varphi_*(\Gamma)$ is the same as the duration of $\Gamma$.
\end{cor}

\noindent\textbf{Comments.} A few considerations are in order:
\begin{enumerate}
\item There can be no ``time'' function $t$ on $M$, because it should hold identically $Dt=1$ for each second order differential equation, or, that is the same, $\dot t=1$, which is absurd. If a function $t$ is chosen as time, we are restricting the manifold  of states $TM$ to the hypersurface $\dot t=1$, by means of a ``time constraint'' (see Section \ref{timeconstraint}).
    \medskip

\item Given the configuration space $M$, a projection $\varphi\colon M\to N$ can be interpreted as a ``reduction of the number of observed degrees of freedom''. That reduction does not modifies the measure of the time, in the precise sense given in Corollary \ref{absolutetime}. In the extreme case of $\dim N=1$, $N$ is a \emph{clock} for $M$.
     \medskip

\item The time-duration is independent of the metric given on $M$.  The confusion with the time-coordinate of Minkowski or the ``proper time'' of Relativity (which is the length with respect to a given metric; see below)  should be avoided.
     \medskip

\item The time as duration is defined only for trajectories in $TM$ which (except for their parametrization) have as tangent accelerations or vertical vectors. When the position ``$a$'' and the velocity ``$v_a$'' are not coupled in that sense, it is not possible to speak about time as duration. As a result, the issue of the nature of time in Quantum Mechanics is strongly non trivial. In the present work we propose two different ways to reach the time dependent Schrödinger equation (for conservative systems) by means of the quantization of classical states. One way, consist of considering a time constraint on a free system. The other way is based on the point of view of Hertz, by considering the system as a projection of a free system. In the later manner the ``time'' appearing in the Schrödinger equation is a linear combination of the time-duration with the ``time'' quotient of the action by the energy on each solution of the Hamilton-Jacobi equation.
     \medskip

\item If the observable Universe were a classical mechanical system $(M,D)$, its story would be a trajectory $\Gamma$ of $D$. The observations would be taken with respect to stretches $\Gamma_i$ of $\Gamma$, projected by means of a process  of ``forgetting degrees of freedom'', $\pi'\colon M\to M'$, along curves $\Gamma_i'$. If $\Gamma_i'$ is outside of the 0-section of $TM'$ (this is to say, if there is no point of $\Gamma'_i$ in which ``everything stops''), the duration of $\Gamma'_i$ is well defined, and it is the same as that of the $\Gamma_i$. Another observation of the same stretch, done by means of $\pi''\colon M\to M''$, will give a curve $\Gamma_i''$ with the same duration. The temporal correlation between such pairs of curves $\Gamma'$, $\Gamma''$, makes unnecessary to observe the whole of the Universe evolution to be able of measure the time.
\end{enumerate}

\section{Classical Mechanics}\label{classical}

In this section we will review the Classical Mechanics in the approach presented in \cite{MecanicaMunoz} and \cite{RM}.
In this approach, we introduced the notion of time constraint and  the concept of Hertz constraint was recovered.
We will continue with the notation given in Section \ref{time}.\medskip

Newton Mechanics is based on a link between two objects: forces and accelerations. But, for forces to  produce accelerations it is necessary a riemannian metric (of arbitrary signature) in $M$, as we will see soon.

Let $T^*M$ be the cotangent bundle of $M$. If $\U\subset M$ is an open set coordinated by $\{x^i\}$, we define the functions $p_i$ on $T^*\U\subset T^*M$ by the rule $p_i(\alpha_a):=\alpha_a(\partial/\partial x^i)$ for each $\alpha_a\in T^*\U$. The functions $\{x^i,p_i\}$, $i=1,\dots,n$, are local coordinates on $T^*\U$.

The \emph{Liouville form} $\theta$ is defined by
$$\theta_{\alpha_a}=\pi^*\alpha_a,\quad \text{for each}\,\,\alpha_a\in T^*M$$
where $\pi^*$ denotes the pullback by the projection $T^*M\to M$.
Its exterior differential, $d\theta=\omega_2$, is the canonical \emph{symplectic form} in $T^*M$. Their expressions in local coordinates are $\theta=p_i\,dx^i$ and $\omega_2=dp_i\wedge dx^i$.

Given a non-degenerate metric $T_2$ on $M$ (a 2-covariant symmetric tensor field without kernel, of arbitrary signature), we establishe, in the usual way, an isomorphism $TM\simeq T^*M$. Thus, thanks to the metric $T_2$ on $M$, we will talk about the Liouville form $\theta$ in $TM$, which will be given by $\theta_{v_a}=v_a\lrcorner\,T_2$ (pulled-back from $M$ to $TM$) and, in the whole of $TM$, $\theta={\dot d}\lrcorner\,T_2$. The tautological structures  $\dot d$ in $TM$ and $\theta$ in $T^*M$ correspond to each other by $T_2$.

 If, in local coordinates,  $T_2=g_{ij}\,dx^i dx^j$, the isomorphism between $TM$ and $T^*M$ is expressed by
  $$p_i=g_{ij}\dot x^j\quad\text{or}\quad \dot x^j=g^{ij}p_i,$$
where $g^{ij}$ denote the entry $(i,j)$ of the inverse matrix of $(g_{ij})$. In particular,
  $$\theta=g_{ij}\dot x^i\,dx^j$$
 in $TM$.
  The function $T:=(1/2)\dot\theta=(1/2)g_{ij}\dot x^i\dot x^j$ is the \emph{kinetic energy}.

 In order to avoid obvious precisions, we will go from $TM$ to $T^*M$ and viceversa, except if there is a risk of confusion, assuming that we will use the isomorphism established by the metric.
 \medskip

 The entirety of the Mechanics rests on the following
\begin{lemma}[Fundamental Lemma of Classical Mechanics]
The metric $T_2$ establishes a univocal correspondence between second order differential equations $D$ on $M$ and horizontal 1-forms $\alpha$ in $TM$, by means of the following equation
\begin{equation}\label{Newton}
D\lrcorner\,\omega_2+dT+\alpha=0.
\end{equation}
The tangent fields $u$ on $M$ that are intermediate integrals of $D$ are precisely those holding
\begin{equation}\label{NewtonIntermedia}
u\lrcorner\, d(u\lrcorner\,T_2)+dT(u)+u^*\alpha=0, 
\end{equation}
where $u^*\alpha$ is the pull-back of $\alpha$ by means of the section $u\colon M\to TM$ and $T(u)$ is the function $T$ specialized to $u$.
\end{lemma}
\begin{proof}
\cite{MecanicaMunoz,RM}

(\ref{Newton}) Given the second order differential equation $D$, we define the form $\alpha$ by the rule (\ref{Newton}). We must check that $\alpha$ is horizontal, that is to say, that $\langle\alpha,V\rangle=0$ for all vertical field $V$.

By using the classical formula of Cartan, we have
$$\langle D\lrcorner\,\omega,V\rangle =\langle D\lrcorner\,d\theta,V\rangle=D\langle\theta,V\rangle-V\langle\theta,D\rangle-\langle\theta,[D,V]\rangle.$$
Since $\theta$ is horizontal, we have $\langle\theta,V\rangle=0$. For the same reason, $\langle \theta,D\rangle=\langle\theta,\dot d\rangle=\dot\theta=2T$ and, since $\dot\theta$ is homogeneous of second degree in the $\dot x$',
$V\dot\theta=2\langle\theta,v\rangle$, where $v$ is the geometric representative of $V$. Therefore, we have $V\langle\theta,D\rangle=2\langle\theta,v\rangle$.

For $f\in\A$ it holds
$$[D,V]f=-VDf=-V\dot f=-vf.$$
It is derived that  $\langle\theta,[D,V]\rangle=-\langle\theta,v\rangle$. Putting it all together,
$$\langle D\lrcorner\,\omega_2,V\rangle=0-2\langle\theta,v\rangle+\langle\theta,v\rangle=-\langle\theta,v\rangle.$$
In addition,
$$\langle dT,V\rangle=VT=\frac 12 V\dot\theta=\langle\theta,v\rangle.$$

Adding up, it follows $\langle\alpha,V\rangle=0$, and then $\alpha$ is proved to be horizontal.

Since $\omega_2$ has no radical, the correspondence $D\to\alpha$ is injective. On the other hand, if $V$ is any vertical field, $D+V$ is a second order differential equation and, from the above proved equalities, it results that $V\lrcorner\,\omega_2$ is horizontal. Since vertical fields and horizontal 1-forms are locally free $\C^\infty(TM)$-modules  with the same rank, it follows that, given an arbitrary horizontal form $\alpha$, there exists a $D$ that holds (\ref{Newton}).
\medskip

(\ref{NewtonIntermedia}) Thinking of $u$ as a section of $TM\to M$, the fact of being an intermediate integral of $D$ means that $D$ is tangent to $u$. The specialization of $\theta$ to the section $u$ is $u\lrcorner\,T_2$ (lifted to the section $u$) and, therefore, that of $\omega_2$ is $d(u\lrcorner\,T_2)$. Thus, the condition (\ref{NewtonIntermedia}) on $u$ means that the specialization  of the 1-form ${u_*u}\lrcorner\,\omega_2$ coincides with that of $-dT(u)-\alpha$  which, by (\ref{Newton}), is the specialization of $D\lrcorner\,\omega_2$. Consequently, the condition (\ref{NewtonIntermedia}) means that the vertical field $V=D-u_*u$ (supported on $u$) holds $V\lrcorner\,\omega_2|_u=0$ and, being $V\lrcorner\,\omega_2$ horizontal, it is equivalent to $V=0$.
\end{proof}

The local expression for the second order differential equation
$$D=\dot x^j\frac{\partial}{\partial x^j}+\ddot x^j\frac{\partial}{\partial\dot x^j},\qquad \text{where}\quad \ddot x^i:=D\dot x^i,$$
determined by the horizontal form $\alpha=\alpha_j(x,\dot x)\,dx^j$ according the above theorem is
\begin{equation}\label{Newtoncoordenadas}
g_{i j}\ddot x^j+\Gamma_{k\ell,i}\,\dot x^k\dot x^\ell+\alpha_i=0
\end{equation}
or, by putting $\alpha^j=g^{ij}\alpha_i$,
\begin{equation}\label{Newtoncoordenadas2}
\ddot x^j+\Gamma_{k\ell}^j\,\dot x^k\dot x^\ell+\alpha^j=0.
\end{equation}

\begin{obs} In the proof of (\ref{Newton}) we have seen that, for each vertical field $V$, the 1-form $V\lrcorner\,\omega_2$ is horizontal. We have, in fact, the formula
\begin{equation}\label{Verticalgeometrico}
V\lrcorner\,\omega_2=v\lrcorner\,T_2,
\end{equation}
where $v$ is the geometric representative of $V$. The proof can be easily done in local coordinates.
\end{obs}

  \begin{defi}
  A \emph{classical mechanical system} is a manifold $M$ (the \emph{configuration space}) endowed with a riemannian metric (of arbitrary signature and nondegenerate) $T_2$ and a horizontal 1-form $\alpha$, the \emph{form of work} or \emph{form of force}. The second order differential equation $D$ corresponding with $\alpha$ by (\ref{Newton}) is the \emph{differential equation of the motion} of the system $(M,T_2,\alpha)$ and (\ref{Newton}) is the \emph{Newton equation}.
  \end{defi}

  For $\alpha=0$ we have the \emph{system free of forces} or \emph{geodesic system}, whose equation of motion $D_G$ is the \emph{geodesic field}: ${D_G}\lrcorner\,\omega_2+dT=0$. $D_G$ is the Hamiltonian field corresponding to the function $T$ by means of the symplectic form $\omega_2$ associated with the given metric.

  The geodesic field provides an origin for the affine bundle of the second order differential equations. For each second order equation $D$, $D-D_G=V$ is a vertical field, the \emph{force} of the system $(M,T_2,\alpha)$ whose equation of the motion is $D$. The force $V$ and the form of force $\alpha$ are related  by
  $$V\lrcorner\,\omega_2+\alpha=0,\quad\text{or}\quad v\lrcorner\,T_2+\alpha=0\quad\text{or}\quad v=-\textrm{grad}\,\alpha,$$
 where $v$ is the geometric representative of $V$. The geometric representative of $V=D-D_G$ will be called \emph{covariant value of $D$} with respect to the metric $T_2$, and we will denote it by $D^\nabla$. The reason of its name is the following: if $\gamma$ is a parameterized curve in $M$ that is solution of a second order differential equation $D$, and $u$ is the field tangent to the curve, then $\nabla_uu=D^\nabla$ at each point of $\gamma$ (see formula (22) in \cite{MecanicaMunoz}). The equation
 \begin{equation}\label{NewtonCovariante}
 {D^\nabla}\lrcorner\,T_2+\alpha=0\quad\text{or}\quad D^\nabla=-\textrm{grad}\,\alpha\quad\text{or}\quad \nabla_uu=-\textrm{grad}\,\alpha\,\,\,\text{(along $\gamma$)}
 \end{equation}
 is a form of the Newton equation (\ref{Newton}) closer to the old ``force = mass $\times$ acceleration''.

\subsection{Conservative systems}\label{subsectionconservatives}

When $\alpha$ is an exact differential form, the system $(M,T_2,\alpha)$ is said to be \emph{conservative}. Since $\alpha$ is horizontal, the potential function $U$, such that $\alpha=dU$, belongs to $\A$. In this case, equations (\ref{Newtoncoordenadas}) are
\begin{equation}\label{Newtoncoordenadas3}
g_{i j}\ddot x^j+\Gamma_{k\ell,i}\,\dot x^k\dot x^\ell+\frac{\partial U}{\partial x^i}=0.
\end{equation}

The sum $H:=T+U$ is called \emph{Hamiltonian} of the system. Newton equation  (\ref{Newton}) becomes in this case,
\begin{equation}\label{Newtonconservativo}
D\lrcorner\,\omega_2+dH=0.
\end{equation}

When (\ref{Newtonconservativo}) is written in coordinates $(x,p)$ of $T^*M$, we get the system of Hamilton \emph{canonical equations}. The specialization of $\omega_2$ into each hypersurface $H=\textrm{const.}$ of $T^*M$,  has as radical the field $D$ (and its multiples), as it is derived from (\ref{Newtonconservativo}). From the classical argument based in the Stokes theorem it results, then, the \emph{Maupertuis Principle}: the curves in $M$ with given end points that, lifted to $TM$, remain within the same hypersurface $H=\textrm{const.}$, give values for $\int\theta$ which are extremal just for the trajectory of the system.

The equation (\ref{NewtonIntermedia}) for the intermediate integrals of $D$ is, in this case:
\begin{equation}\label{NewtonconservativoIntermedia}
u\lrcorner\,d(u\lrcorner\,T_2)+dH(u)=0.
\end{equation}

In particular, if $u$ is a Lagrangian submanifold of $TM$, this is to say, if $d(u\lrcorner\,T_2)=d\theta|_u=0$, the equation (\ref{NewtonIntermedia}) is
$$dH(u)=0\quad\text{or}\quad H(u)=\textrm{const.}$$
This is the \emph{Hamilton-Jacobi equation}, when is translated to the function $S$ of which $u$ is the gradient:
\begin{equation}\label{Hamilton-Jacobi}
H(\textrm{grad}\,S)=\textrm{const.}\,\,\text{in $TM$, or}\quad H(dS)=\textrm{const.} \,\,\text{in $T^*M$,}
\end{equation}
or, in coordinates,
$$\frac 12\,g^{jk}\,\frac{\partial S}{\partial x^j}\frac{\partial S}{\partial x^k}+U=\textrm{const.}$$

This is  a first order partial differential equation, with one unknown that does not  appear explicitly in the equation (only its derivatives appear). The local theory of this equations was developed by Jacobi \cite{Jacobi} and clarified  and completed by Lie in the 1870' \cite{Engel-Faber} in connection with the theory of contact transformations.

\subsection{Time constraints}\label{timeconstraint}

We can   \emph{impose a time} for $M$ by choosing a horizontal 1-form $\tau$ in $TM$ and admitting as position-velocity states only the points of $TM$ where it holds $\dot \tau=1$. This process  can be called a \emph{time constraint}. In particular, for $\tau=dt$, where $t$ is a function in $M$, the constraint $\dot t=1$ selects among the trajectories of second order differential equations those where $t$ ``flows equably'', as the absolute time of Newton.

In general, the field $D$ that governs the evolution of the mechanical system $(M,T_2,\alpha)$ is not tangent to the manifold $\dot\tau=1$ of the time constraint, and such a constraint imposes the modifying $D$ by changing it by other field $\overline D$ in a manner quite analogous to ordinary constraints. $\overline D$ must hold the conditions
\begin{align}
& {\overline D}\lrcorner\,\omega_2+dT+\alpha \equiv 0\,\, (\textrm{mod}\, \tau)\label{campotiempo1} \\ 
& \overline D\dot\tau=0, \qquad \text{($\overline D$ is tangent to the manifolds $\dot \tau=\textrm{const.}$).}\label{campotiempo2}
\end{align}

The above conditions completely determine the field $\overline D$ in the open set $\|\tau\|\ne 0$. In fact, 
\begin{equation}\label{campotiempo3}
\overline D=D-\frac{D\dot\tau}{\|\tau\|^2}\,\textrm{Grad}\,\tau 
\end{equation}
($\textrm{Grad}\,\tau$ is the vertical representative of $\textrm{grad}\,\tau$, and holds ${\textrm{Grad}\,\tau}\lrcorner\,\omega_2=\tau$, according (\ref{Verticalgeometrico})).

The field $\overline D$ when is restricted to the ordinary constraint $\dot\tau=0$, equals the field that governs the evolution of the system $(M,T_2,\alpha)$ subjected to that constraint.

In \cite{MecanicaMunoz}, section 3.12, it is proved that, when $\tau$ is a 1-form on $M$ (so that, $\dot\tau=0$ is, then, a linear constraint), (\ref{campotiempo3}) is equivalent to
\begin{equation}\label{campotiempo4}
\overline D=D-\frac 1{\|\tau\|^2}\left(\langle\tau,D^\nabla\rangle+\textrm{II}_{\textrm{grad}\,\tau}(\dot d,\dot d)\right)
\,\textrm{Grad}\,\tau 
\end{equation}
where $\textrm{II}_u$ is the second fundamental form of the field $u$ with respect to the metric $T_2$.

Formulae (\ref{campotiempo3}), (\ref{campotiempo4}) show that the enforce of $\dot\tau=1$ as temporal evolution law is done by imposing an additional force on the given system.

As an example, let us see that any conservative system is projection of a free system (geodesic) to which it is imposed a time constraint:

Let $\M$ be a manifold of dimension $n+1$, with metric $\T_2$. Let $x^0$ be a function in $\M$ whose differential has no zeroes (then, $x^0$ is said to be regular. We can restrict $\M$ to the open set where this condition holds, if necessary). Let us take as local coordinates in $\M$ the function $x^0$ along with $n$ first integrals $x^\mu$ ($\mu=1,\dots,n$) of $\textrm{grad}\,x^0$. Thus, $\T_2$ take the form
\begin{equation}\label{campotiempo5}
\T_2=g_{00}\,(dx^0)^2+g_{\mu\nu}\,dx^\mu dx^\nu\,\quad (\mu,\nu=1,\dots,n) 
\end{equation}

Let us suppose, moreover, that $\T_2$ is projectable to the ring of first integrals of $\textrm{grad}\,x^0$. That is to say, when $f$, $g$ are first integrals of $\textrm{grad}\,x^0$, also so is $\T^2(df,dg)$. With the coordinates we are using, that condition means that the coefficients $g_{\mu\nu}$ do not depend on $x^0$. Finally, let us suppose that also $g_{00}$ is independent of $x^0$ (this condition can be intrinsically expressed as $\textrm{II}_{\textrm{grad}\,x^0}(\textrm{grad}\,x^0,\textrm{grad}\,x^0)=0$). With this conditions, the Christoffel symbols $\Gamma_{ij,k}$ with just one index or all the three indexes 0, vanish. The equations for the geodesic field $\DD_G$ in $T\M$ are
\begin{equation}\label{campotiempo6}
\begin{cases}
g_{00}\,\ddot x^0+2\Gamma_{0\mu,0}\,\dot x^0\dot x^\mu=0\\[0.3em]
g_{\mu\nu}\ddot x^\nu+\Gamma_{\sigma\rho,\mu}\,\dot x^\sigma\dot x^\rho+\Gamma_{00,\mu}(\dot x^0)^2=0
\end{cases} 
\end{equation}
(take (\ref{Newtoncoordenadas}) with $\alpha\equiv 0$ and metric (\ref{campotiempo5})).

Let us impose on the geodesic system in $\M$ the constraint of time $\dot x^0=1$. The geodesic field $\DD_G$ is modified according (\ref{campotiempo3}) to a field $\overline \DD$ which differs from $\DD_G$ in a multiple of $\textrm{Grad}\,x^0=g^{00}\partial/\partial\dot x^0$, for which the coefficients of the $\partial/\partial\dot x^\mu$ ($\mu=1,\dots,n$) both in $\DD_G$ and $\overline\DD$ are equal. The second group of equations (\ref{campotiempo6}) in the manifold $\dot x^0=1$ becomes:
\begin{equation}\label{campotiempo7}
g_{\mu\nu}\ddot x^\nu+\Gamma_{\sigma\rho,\mu}\,\dot x^\sigma\dot x^\rho-\frac 12\DP{g_{00}}{x^\mu}=0 
\end{equation}

If we call $M$ the manifold obtained as projection of $\M$ by the field $\textrm{grad}\,x^0$ (so, $M$ is the manifold of trajectories of $\textrm{grad}\,x^0$), endowed with the metric $T_2$, projection of $\T_2$, we see that (\ref{campotiempo7}) are the equations of evolution in the conservative system $(M,T_2,dU)$, with $U=-(1/2)g_{00}$ (see equations (\ref{Newtoncoordenadas3})). Thus, the conservative mechanical system is the projection of a geodesic system (free of forces) in a configuration space of greater dimension on which we have imposed a time constraint.

\subsection{Hertz constraints}\label{Hertz}
Coming back to the equations  of the geodesics in $\M$ (\ref{campotiempo6}), we observe that the first row is the equation $\DD_G(g_{00}\dot x^0)=0$: $p_0$ is a first integral of $\DD_G$. The second group of equations is
\begin{align*}
g_{\mu\nu}\ddot x^\nu+\Gamma_{\sigma\rho,\mu}\,\dot x^\sigma\dot x^\rho-\frac 12(\dot x^0)^2\DP{g_{00}}{x^\mu}&=0,\quad\text{or}\\[0.8em]
g_{\mu\nu}\ddot x^\nu+\Gamma_{\sigma\rho,\mu}\,\dot x^\sigma\dot x^\rho+\frac 12 \,\,\,p_0^2\,\,\,\DP{g^{00}}{x^\mu}&=0
\end{align*}

On each hypersurface $p_0=P_0$ (constant) of $T\M$, this system of equations projects to $TM$ as a conservative system $(M,T_2,U)$, with $U=(1/2) P_0^2g^{00}$.

The kinetic energy in $T\M$:
$$\T=\frac 12 g^{00}p_0^2+\frac 12 g^{\mu\nu}p_\mu p_\nu$$
restricted to the manifold $p_0=P_0$ is
\begin{equation}\label{campotiempo8}
H=\frac 12 g^{00}P_0^2+T,
\end{equation}
the Hamiltonian of the conservative system $(M,T_2,U)$.

A portion of the kinetic energy in $\M$ is ``transferred'' or ``appears as'' potential energy in $M$.

The idea of considering any mechanical system as the ``observable part'' of a greater geodesic system is the principle of the Hertz Mechanics \cite{Hertz}, of which a short but very illustrative recension is made by Sommerfeld \cite{SommerfeldMecanica}. This is why  we will call \emph{Hertz constraint} the above mentioned type. That allows us to project a system to another one of lower dimension by using first integrals of the corresponding equations of motion. Our present case was studied for R. Liouville \cite{RLiouville,Lutzen} prior to Hertz.

As an example, (\ref{campotiempo8}) shows that newtonian gravitation can be represented as a Hertz constraint in a 4-dimensional space with a metric of Minkowskian signature: $g^{00}=-\textrm{const.}/r$, $T_2$ positive definite.

The Hertz constraint that  we have considered is a non-holonomic time constraint, with the form $\tau=g_{00}\,dx^0$, in the manifold $\dot\tau=P_0$. Being $\dot\tau$ first integral of $\DD_G$, the field $\overline\DD$ in (\ref{campotiempo3}) is again $\DD_G$, now specialized to the hypersurface $\dot\tau=P_0$. The equations $\dot\tau=\textrm{const.}$ give a foliation of $T\M$ by hypersurfaces, each one of them giving in $(M,T_2)$ a potential energy, energies that differ from each other by constant factors.
\bigskip

\section{Quantization of contravariant tensors. Dequantization of differential operators}\label{tres}

We will study the way in which a symmetric linear connection $\nabla$ on the configuration space $M$ determines a canonical biunivocal correspondence (up to the concrete value of $h$) between  contravariant tensor fields on $M$ and linear differential operators acting on $\A$. The passage
$\text{\emph{tensor}}\to\text{\emph{differential operator}}$
is the rule of quantization defined by $\nabla$ and the reverse step $\text{\emph{differential operator}}\to\text{\emph{tensor}}$ is the rule of dequantization defined by $\nabla$; this second passage, once given, can be continued with another one
$\text{\emph{tensor}}\to\text{\emph{infinitesimal}}$  \emph{canonical transformation on $T^*M$}, which  already only depends on the structure of $T^*M$.

This section is related with \cite{LychaginQ}. The procedure of quantization is equivalent to the proposed in \cite{QuantizacionMunozAlonso} as it is proved in \cite{QuantizacionTensor}.

The functions that we consider now are of class $\C^\infty$ with values in $\mathbb{C}$.

The quantization rule established with the data $(M,\nabla)$ is an almost obvious generalization of the usual rule of quantization on the flat space $(\R^n,d)$, where we denote by $d$ the connection canonically associated with the vector structure of $\R^n$ (the ``parallel transport'' for $d$ is the transport by linear parallelism). Let us recall such a rule.

Let $E$ be  a real $n$-dimensional vector space. Once fixed a system of vector coordinates $(x^1,\dots,x^n)$ on $E$, each real symmetric contravariant tensor of order $r$ at the origin of $E$ is written in the form:
$$\Phi_0=a^{j_1\cdots j_r}\left(\frac\partial{\partial x^{j_1}}\right)_0\cdots\left(\frac\partial{\partial x^{j_r}}\right)_0,$$
where the $a$ are real numbers and by $\cdots$ we denote the symmetrized tensor product.

The linear structure of $E$ allows us to propagate ``by parallelism''  the tensor $\Phi_0$ to a tensor field $\Phi$ on the whole of $E$, whose expression is the same as that of $\Phi_0$, by deleting the subindex $0$. This tensor field $\Phi$ defines on $\CE$ a differential operator
$$\widehat\Phi:=(-i\hbar)^ra^{j_1\cdots j_r}\frac{\partial^r}{\partial x^{j_1}\cdots\partial x^{j_r}}.$$
The assignation $\Phi\to\widehat\Phi$ is independent of changes of \emph{vector} coordinates on $E$. When  the numbers $a^{j_1\cdots j_r}$ are substituted by functions in $\CE$, the same formula assigns to the  tensor field $\Phi$ a differential  operator $\widehat\Phi$ independently of the concrete choice of \emph{vector} coordinates.

In we wish that $\widehat\Phi$ to be self-adjoint (for the measure translation invariant of $E$) it is sufficient to replace it by $\frac 12(\widehat\Phi+\widehat\Phi^+)$.

When we work on a concrete problem in curvilinear coordinates, the quantization rule is applied by passing the tensors to vector coordinates, quantizing them according to the above rule and, then, coming back to the given curvilinear coordinates.

This recipe for quantization is \emph{intrinsically determined} by the vector structure of $E$. On each $f\in\CE$ we have
$$\widehat\Phi(f)=(-i\hbar)^r\langle\Phi,d^rf\rangle,$$
where $\langle\,,\,\rangle$ denotes tensor contraction, and $d^rf$ is the $r$-th iterated differential of $f$, that has an intrinsic sense on $E$ because of its vector structure.

The generalization to any smooth manifold $M$ endowed with a symmetric (=torsionless) linear connection $\nabla$ is immediate:
\begin{defi}[Quantization defined by $\nabla$]
 For each symmetric contravariant tensor field of order $r$, $\Phi$, on $(M,\nabla)$, the \emph{quantized of $\Phi$} is the differential operator $\widehat\Phi$ which, for each  $f\in\A$ gives
\be\label{quantizado}
\widehat\Phi(f):=(-i\hbar)^r\langle\Phi,\nabla^r_{\textrm{sym}}f\rangle
\ee
where $\langle\,,\,\rangle$ denotes tensor contraction and $\nabla^r_{\textrm{sym}}f$ is the symmetrized tensor of the $r$-th covariant iterated differential of $f$ with respect to the connection $\nabla$.

The quantized of a non-homogeneous tensor is the sum of the quantized of its homogeneous components.
\end{defi}

\begin{obs}
This definition can be generalized giving a differential operator between sections of fibre bundles for each contravariant tensor $\Phi$ on $M$, once a linear connection is fixed in the first fibre bundle. This generalization does not affect what follows, and we leave it aside.
\end{obs}

Let us recall that a differential operator of order $r$ on $M$ (= differential operator of order $r$ on $\A$) is an $\Com$-linear map $P\colon\A\to\A$ which holds the following condition: for each point $x\in M$, $P$ takes the ideal $\m_x^{r+1}$ into $\m_x$ ($\m_x$ is the ideal of the functions of $\A$ vanishing at $x$).

It is derived that $P$ takes the quotient $\m_x^{r}/\m_x^{r+1}$ into $\A/\m_x=\Com$. By taking into account that $\m_x^{r}/\m_x^{r+1}$ is the space of symmetric covariant tensors of order $r$ at the point $x$ (homogeneous polynomials of degree $r$, with coefficients in $\Com$, in the $d_xx^1,\dots,d_xx^n$, once taken local coordinates), we see that $P$ determines a symmetric contravariant tensor of order $r$
 called \emph{symbol of order $r$ of $P$ at $x$}, denoted by $\sigma_x^r(P)$,
\be\label{simbolo}
\sigma_x^r(P)\colon\m_x^{r}/\m_x^{r+1}=T^{*r}_xM\to\R,
\ee
that is the map canonically associated with $P$ by pass to the quotient.
\begin{obs}
Given $f\in\m_x^r$, the differential operator $P$ of order $r$ gives $(Pf)(x)$, depending only on the class $[f]_{\textrm{mod}\,\m_x^{r+1}}$. But the identification of $\m_x^r/\m_x^{r+1}$ with the space of symmetric covariant tensors of order $r$ at the point $x$ is not unique. In order to fix the tensor $\sigma_x^r(P)$ in such a way that its contraction with the symmetric covariant tensor that represents $[f]_{\textrm{mod}\,\m_x^{r+1}}$, to be $(Pf)(x)$, we take such a covariant tensor as $d_x^rf$, computed in any local system of coordinates; the covariant tensor $d_x^rf$ so calculated for $f\in\m_x^r$, does not depends on the choice of coordinates.
\end{obs}

When $x$ runs over $M$, we get the tensor field $\sigma^r(P)$ on $M$ called \emph{symbol of order $r$ of $P$}. If $\sigma^r(P)=0$, $P$ is of order $r-1$.

In the case $M=\R^n$, with vector coordinates $x^1,\dots,x^n$, let us denote $\partial^\alpha$ the tensor
$\partial^\alpha:=(\partial/\partial x^1)^{\alpha_1}\cdots(\partial/\partial x^n)^{\alpha_n}$. Its quantized by the rule (\ref{quantizado})  (with the vector connection of $\R^n$) is $\widehat{\partial^\alpha}:=(-i\hbar)^{|\alpha|}D^\alpha$, where $D^\alpha$ is the differential operator ${\partial^{|\alpha|}}/{(\partial x^1)^{\alpha_1}\dots(\partial x^n)^{\alpha_n}}$. It is directly seen that $\sigma^{|\alpha|}(D^\alpha)=\partial^\alpha$, so that for any tensor field of order $r=|\alpha|$ on $\R^n$ is obtained, by adding terms,
\be\label{simboloplano}
\sigma^r\left(\widehat\Phi\right)=(-i\hbar)^r\Phi
\ee

Going from $\R^n$ to the general case $(M,\nabla)$ let us observe that, when the iterated differentials of a function $f$ are calculated in local coordinates, the derivatives of order $r$ of $f$ appear in terms which does not contain Christoffel  symbols (as in the case of $\R^n$). Since the symbol of an operator of order $r$ depends only on these terms, Formula (\ref{simboloplano}) is still valid in general for the quantization rule
(\ref{quantizado}) on $(M,\nabla)$.

\begin{thm}\label{tquantizado1}
The rule of quantization (\ref{quantizado}) establishes a biunivocal correspondence between linear differential operators $P$ and symmetric contravariant tensor fields (not necessarily homogeneous) on $M$. To the operator $P$ of order $r$ corresponds the tensor $\Phi=\Phi_{r}+\Phi_{r-1}+\cdots\Phi_{0}$ (each $\Phi_j$ denotes the homogeneous component of degree $j$) such that
$$
\sigma^r(P)=(-i\hbar)^r\Phi_r
$$
and, for $k=1,\dots,r$:
$$\sigma^{r-k}(P-\widehat\Phi_{r}-\cdots-\widehat\Phi_{r-k+1})=(-i\hbar)^{r-k}\Phi_{r-k}$$
and
\be\label{quantizado3}
P=\widehat\Phi=\widehat\Phi_{r}+\widehat\Phi_{r-1}+\cdots+\widehat\Phi_{0}.
\ee
\end{thm}

\begin{defi}[Dequantization]
The contravariant tensor $\Phi$ in (\ref{quantizado3}) is the \emph{dequantized of the differential operator $P$ by the connection $\nabla$}.
\end{defi}

Symmetric contravariant tensor fields (homogeneous or not) on $M$ canonically correspond with functions $f\in \AC$ polynomials along the fibres.

\begin{defi}
The function $F\in\AC$ corresponding to the tensor $\Phi$ dequantized of the differential operator $P$ will be called \emph{Hamiltonian of $P$} with respect to the connection $\nabla$.
\end{defi}

The symplectic structure $\omega_2$ of $T^*M$ assign to each $F\in\AC$  a Hamiltonian vector field $D_F$ by the rule $D_F\lrcorner\,\omega_2+dF=0$. These Hamiltonian fields are the \emph{infinitesimal canonical transformations} of Lie \cite{Lie1,Lie2}; they are the infinitesimal generators of the (local) 1-parametric groups of automorphisms of $T^*M$ which preserve its symplectic structure.

\begin{defi}[Hamiltonian field associated with a differential operator] We will call \emph{infinitesimal canonical transformation associated with the differential operator $P$ or Hamiltonian field associated with $P$} to the tangent field $D_P$ on $T^*M$ such that
$${D_P}\lrcorner\,\omega_2+dF=0,$$
where $F$ is the Hamiltonian of $P$.
\end{defi}

The path $P\to F\to D_P$ is univocal. The reverse path $D_P\to F$ determines $F$ up to a additive constant; then, $F\to P$ is univocal. Thus, up to an additive constant for $P$, the correspondence $P\leftrightarrow D_P$ is biunivocal.
\begin{thm}\label{tquantizado2}
The symmetric linear connection $\nabla$ on $M$ canonically establishes a biunivocal correspondence between linear differential operators $P$ on $\A$ (up to additive constants) and infinitesimal canonical transformations of the simplectic manifold $T^*M$ corresponding to functions polynomial along fibres (Hamiltonians).
\end{thm}

 Let us assume that $M$ is endowed with a Riemannian metric $T_2$ (of arbitrary signature = pseudoriemannian metric) and $\nabla$ the associated Levi-Civita connection. Under these conditions, it makes sense to say whether or not a tangent field $D$ on $T^*M$ is a second-order differential equation; that is, the tangent field that governs a mechanical system with the configuration space $(M,T_2)$. We have,

 \begin{thm}\label{unicosbuenos}
 The necessary and sufficient condition for a linear differential  operator $P$ on $(M,T_2)$ to have as associated infinitesimal canonical transformation $D_P$ a second order differential equation is that $P$ is of the form
 $$P=-\frac{\hbar^2}2\Delta+U$$
 where $\Delta$ is the Laplacian operator of the metric and $U\in\A$.
 \end{thm}
 \begin{proof}
 Let us begin by checking that the tensor $\Phi$, contravariant form of the metric tensor, has as quantized operator $\widehat\Phi=-\hbar^2\Delta$. In local coordinates, with
 $T_2=g_{jk}dx^jdx^k$, is $\Phi=g^{rs}\frac{\partial}{\partial x^r}\otimes\frac{\partial}{\partial x^s}$. The expression for the second iterated covariant differential is
 $$\nabla^2f=\left(\frac{\partial^2f}{\partial x^k\partial x^j}-\Gamma_{jk}^\ell\frac{\partial f}{\partial x^\ell}\right)dx^j\otimes dx^k;$$
 by contracting with $\Phi$,
 $$\langle\Phi,\nabla^2f\rangle=g^{jk}\left(\frac{\partial^2f}{\partial x^k\partial x^j}-\Gamma_{jk}^\ell\frac{\partial f}{\partial x^\ell}\right)=\Delta f.$$
 By incorporating to $\Phi$ the factor $(-i\hbar)^2$ we see that the quantized of $\Phi$ is $-\hbar^2\Delta$.

 The Hamiltonian function corresponding to the tensor $\Phi$ is $g^{rs}p_rp_s=2T$ (where $T$ is the kinetic energy function). Finally, for the Hamiltonian $H=T+U$, the corresponding quantum operator is $(-\hbar^2/2)\Delta+U$.

 When dequantizing, we go from the operator $(-\hbar^2/2)\Delta+U$ to the Hamiltonian $T+U=H$, and, then to the Hamiltonian field $D_P$ such that $D_P\lrcorner\,\omega_2+dH=0$; $D_P$ is the field of the canonical equations for the mechanical system $(M,T,dU)$.

 Conversely, let us assume that $D_P$ is a second order differential equation. Equation \ref{Newton} gives that it holds $D_P\lrcorner\,\omega_2+dT+\alpha=0$, where $\alpha$ is horizontal; since $D_P$ is a canonical infinitesimal transformation, $\alpha$ has  to be exact, so that of the form $dU$ for some $U\in\A$: $D_P\lrcorner\,\omega_2+dH=0$, for $H=T+U$. Since $T$ is the Hamiltonian function associted with the operator $\frac 12\Phi$ as before, when quantizing it turns that $P=-\frac{\hbar^2}2\Delta+U$.
  \end{proof}

\noindent\textbf{Problem.} There is something similar to a Schrödinger equation for general non-conservative mechanical systems?
\medskip

\begin{obsparticulares}
 Let $\Phi_r$ be the homogeneous tensor of order $r$ that corresponds to $P$ by (\ref{quantizado3}). Considered as a function on $T^*M$, $\Phi_r$ is $F_r$, homogeneous of degree $r$ on the fibres. The first order partial differential equation $F_r((dS)^r)=0$ has as solutions the hypersurfaces $S=\text{const.}$ \emph{characteristic} for the differential operator $P$; they are the hypersurfaces of $M$ where the problem of initial conditions cannot be treated by the Cauchy-Kowalevsi method (for instance, for $\Delta$, the equation of characteristics is $\|dS\|^2=0$, the ``Eikonal equation''). The Hamiltonian field of $F_r$ has as solutions the \emph{bicharacteristics} of $P$. This field does not coincide, in general, with $D_P$. The field which propagates the singularities of $P$ is the Hamiltonian field of $F_r$, not the one of the total Hamiltonian of $P$, $D_P$.
\end{obsparticulares}

\section{Time dependent Schrödinger equation from a time constraint}

This is an approach to the time dependent Schrödinger equation whose point of departure will be a time constraint (see Section \ref{timeconstraint}).
\medskip

\subsection{The Hamilton-Jacobi equation for a time constraint} To begin with, let us recover the time-dependent Hamilton-Jacobi equation in this framework. Let $(\M,\T_2)$ be  an $(n+1)$-dimensional riemannian manifold and $\DD$ its geodesic field.

Now we fix the (only needed) additional datum: let us fix a regular function $t$ on $\M$.

Let us consider the time constraints $\dot t=c$ on the geodesic system given by $\DD$. The result is a field $\overline\DD$ (defined all over $T\M$) which holds  the system of equations (\ref{campotiempo1})-(\ref{campotiempo2}), now specified as
\begin{align}
& {\overline \DD}\lrcorner\,d\theta+d\T\equiv 0\,\, (\textrm{mod}\, dt)\label{Otrocampotiempo1} \\
& \overline \DD\dot t=0.
        \label{Otrocampotiempo2}
\end{align}

From now on, we will restrict ourselves to the open set $\{\dot t\ne 0\}$. Then, if $\lambda$ is the function such that (\ref{Otrocampotiempo1}) becomes the equality
$${\overline \DD}\lrcorner\,d\theta+d\T+\lambda dt=0,$$
we get, by interior product by $\overline\DD$, the expression
$\lambda=-{\overline \DD\T}/{\dot t}.$ Once inserted this value into the above relation, and using $\overline\DD t=\dot t$, we arrive to
\begin{equation*}
 {\overline \DD}\lrcorner\,\left(d\theta-d\T\wedge\frac{dt}{\dot t}\right)=0.
 \end{equation*}
 If now we take into account that $\overline\DD \dot t=0$, a little manipulation shows
\begin{equation}\label{OtroNewtontiempo3}
 {\overline \DD}\lrcorner\,d\overline\theta+\frac \T{\dot t}\,d\dot t=0,
 \end{equation}
 where we have put
 \begin{equation*}
 \overline\theta:=\theta-\frac{\T}{\dot t}dt.
 \end{equation*}
 or, in local coordinates $(t=x^0,x^1,\dots,x^n)$, $$\overline\theta=\left(p_0- \T/\dot t\right)dt+p_1dx^1+\cdots+p_ndx^n.$$

 In fact, (\ref{OtroNewtontiempo3}) replaces the couple of equations  (\ref{Otrocampotiempo1})-(\ref{Otrocampotiempo2}): let us suppose that (\ref{OtroNewtontiempo3}) holds. By interior product with ${\overline\DD}$ we get
 $$0={\overline\DD}\lrcorner\,{\overline\DD}\lrcorner\,d\overline\theta+{\overline\DD}\lrcorner\,\left(\frac\T{\dot t}d\dot t\right)=\frac\T{\dot t}\overline\DD(\dot t),$$
 so that $\overline\DD(\dot t)=0$, which is (\ref{Otrocampotiempo2}). Then,
 \begin{align*}
 0={\overline \DD}\lrcorner\,d\overline\theta+\frac \T{\dot t}\,d\dot t
   &={\overline \DD}\lrcorner\,d\left(\theta-\frac\T{\dot t}dt\right)+\frac \T{\dot t}\,d\dot t
               \\
     &     = {\overline \DD}\lrcorner\,\left(d\theta-d\left(\T/{\dot t}\right)\wedge dt\right)+\frac \T{\dot t}\,d\dot t \\
     &
     ={\overline\DD}\lrcorner\,d\theta-\frac{\overline\DD(\T)}{\dot t}dt+\overline\DD t\,d\left(\T/{\dot t}\right)+\frac \T{\dot t}\,d\dot t \\
     &
     ={\overline\DD}\lrcorner\,d\theta-\frac{\overline\DD(\T)}{\dot t}dt+\dot t\frac{\dot t d\T-\T d\dot t}{(\dot t)^2}+\frac \T{\dot t}\,d\dot t \\
      &={\overline\DD}\lrcorner\,d\theta-\frac{\overline\DD(\T)}{\dot t}dt+ d\T,
 \end{align*}
 so that (\ref{Otrocampotiempo1}) holds.
 \medskip

 In other terms,  $(\dot t/\T)\overline\DD$ is the Hamiltonian field associated to the function $\dot t$, with respect to the symplectic form $d\overline\theta$. That property completely characterizes the field $\overline\DD$.
  \medskip

 Now, we define the time-dependent Hamilton-Jacobi equation to be the partial differential equation determining the $d\overline\theta$-Lagrangian intermediate integrals of $\overline\DD$: locally, such a Lagrangian submanifold is defined by $\overline\theta=dS$, for a function $S\in\C^\infty(\M)$. Then, on the image of $dS\subset T^*\M$, (\ref{OtroNewtontiempo3}) becomes
 $$0={\overline\DD}\lrcorner\,ddS+\frac{\T}{\dot t}d\dot t=\frac{\T}{\dot t}d\dot t.$$
As a consequence, $\dot t=c$ for a constant $c\in\R$.
\medskip

In order to get the explicit Hamilton-Jacobi equation for $S$, let us simplify the presentation by taking a system of coordinates $(x^0=t,x^1,\dots,x^n)$ where the $x^\mu$ are first integrals of $\textrm{grad}\,t$ (or, in other words, $dt$ is orthogonal to $dx^\mu$, $\mu\ge 1$). Thus, we have
 $$\T_2=g_{00}(dt)^2+g_{\mu\nu}dx^\mu dx^\nu,$$
 so that the kinetic energy is
 $$\T=\frac 12 g^{00}p_0^2+\frac 12 g^{\mu\nu}p_\mu p_\nu.$$

 \begin{obs}\label{espaciotiempo}
 Once the function $t$ is fixed, an orthogonal decomposition is induced on each tangent space $T_x\M=\langle\textrm{grad}_xt\rangle\perp T^{\mathcal{S}}_x\M$, where $\langle\textrm{grad}_xt\rangle$ denotes the line spanned by the gradient of $t$ and $T^{\mathcal{S}}_x\M$ denotes its orthogonal complement ($\mathcal{S}$ is for space). So, each vector decomposes as the sum of a \emph{time component} parallel to $\textrm{grad}\,t$ and a \emph{space component}. Such decomposition extends itself to differential forms on $\M$ and also to horizontal differential forms on $T\M$ (for example, to the Liouville form). Also the very metric $\T_2$ can be decomposed as a time component, $g_{00}(dt)^2$, and a space component, $g_{\mu\nu}dx^\mu dx^\nu$. Later on, we will refer to the time and space components in that sense. It must be clear that that such components are univocally determined (after the choice of $t$).
 \end{obs}

 In this coordinates $p_0=g_{00}\dot t$
 and the  time-coefficient of $\overline\theta$ is
 $$p_0-\frac \T{\dot t}=\frac 12\, p_0-\frac 12 \frac{g^{\mu\nu}}{\dot t}p_\mu p_\nu=\frac 12\, g_{00}\dot t-\frac 12 \frac{g^{\mu\nu}}{\dot t}p_\mu p_\nu.$$

 So that, if we put $\overline\theta=\overline p_jdx^j$, the hypersurfaces of constant $\dot t$ can be (locally and) equivalently described as
   \begin{equation}\label{cambiocoordenadastiempo}
   \dot t=c\qquad \Leftrightarrow\qquad \overline p_0+\frac 12 \frac{g^{\mu\nu}}{c}\overline p_\mu \overline p_\nu-\frac 12\, g_{00}c=0
   \end{equation}

   Then, the image of $dS$ is parameterized by $\overline p^j= \partial S/\partial x^j$. Finally, the explicit time-dependent Hamilton-Jacobi equation  becomes
   $$\frac{\partial S}{\partial t}+\frac 12 \frac{g^{\mu\nu}}{c}\frac{\partial S}{\partial x^\mu} \frac{\partial S}{\partial x^\nu}-\frac 12\, g_{00}c=0.$$
    \medskip

  In the case $c=1$ (so that $\dot t=1$ and $t$ plays the ``true'' role of time), the above equation becomes the usual one,
  $$\frac{\partial S}{\partial t}+H(t,x^\mu,\partial S/\partial x^\mu)=0,$$
  where $$H:=T(dS)+U,\quad T(dS):=\frac 12 g^{\mu\nu}\frac{\partial S}{\partial x^\mu} \frac{\partial S}{\partial x^\nu}\quad\text{and}\quad U:=-\frac 12\, g_{00}.$$
  \bigskip

  \subsection{Quantization of the time constraint}
  Let us fix $c=1$ and consider the \emph{space kinetic energy} function:
  $$T:=\frac 12 {g^{\mu\nu}}\overline p_\mu \overline p_\nu$$
  (so that $T$ well could be denoted by ``$\T^\mathcal{S}$'').
  The quantization of $T$ following  Section \ref{tres} does not take us to the ``space Laplacian'' but an additional term will appear.
  According the calculations in the mentioned section, the quantum operator associated with $T$ is:
  $$\widehat T=-\frac{\hbar^2}2 g^{\mu\nu} \left(
              \frac{\partial^2}{\partial x^\mu\partial  x^\nu}-\Gamma_{\mu\nu}^\gamma\DP {}{x^\gamma}-\Gamma_{\mu\nu}^0\DP {}{t}\right)
                 =-\frac{\hbar^2}2 \left(
              \Delta^{\mathcal{S}}-g^{\mu\nu} \Gamma_{\mu\nu}^0\DP {}{t}\right),$$
   where
  $\Delta^{\mathcal{S}}$ denotes the ``Laplacian'' in the space coordinates $x^\mu$ (so that, here, $t=x^0$ is just a parameter):
  $$\Delta^{\mathcal{S}}:=g^{\mu\nu} \left(
              \frac{\partial^2}{\partial x^\mu\partial  x^\nu}-\Gamma_{\mu\nu}^\gamma\DP {}{x^\gamma}\right).$$
  But, by the choice of the coordinate system,
   \begin{equation}\label{invariantetermino}
  g^{\mu\nu}\Gamma_{\mu\nu}^0=-\Ddelta t-\textrm{grad}(t)(\log\sqrt{|g_{00}|}),
  \end{equation}
  as can be derived from $\Ddelta\, t=-g^{ij}\Gamma_{ij}^0$, ($i,j=0,1,\dots,n$) and $\Gamma_{00}^0=(1/2)\textrm{grad}(t)(g_{00})$.
  In this way,
  $$\widehat T=-\frac{\hbar^2}2\left(\Delta^{\mathcal{S}}-\left[\Ddelta +\textrm{grad}\left(\log\sqrt{|g_{00}|}\right)\right](t)\DP{}{t}\right),$$

  For this reason, the quantization of the constraint given by the right side member of (\ref{cambiocoordenadastiempo}) (in the case $c=1$) is
   \begin{equation}\label{pseudoSchrodinger}
   (i\hbar+\kappa\hbar^2)\DP{\Psi}t=\left(-\frac{\hbar^2}{2}\Delta^{\mathcal{S}}+U\right)\Psi,
      \end{equation}
   where
   \begin{equation*}
   \begin{cases}
  U:=-\frac 12\, g_{00},\\[0.5em]
  \kappa:=-\frac 12\left[\Ddelta +\textrm{grad}\left(\log\sqrt{|g_{00}|}\right)\right](t).
   \end{cases}
   \end{equation*}
  \medskip

  Finally, let us assume that $\T_2$ is projectable by $\textrm{grad}\,t$ (that is, every $g_{\mu\nu}$, $\mu,\nu=1,\dots,n$, is a first integral of $\textrm{grad}\,t=g^{00}\partial/\partial t$) and, in addition, $g_{00}$ is also a first integral of $\textrm{grad}\,t$. Then, (\ref{invariantetermino}) equals $g^{\mu\nu}(-1/2)g^{00}\partial{g_{\mu\nu}}/\partial t$, in such a way that, under the new assumption, $\kappa=0$ and
  (\ref{pseudoSchrodinger}) becomes
    $$i\hbar\DP{\Psi}t=\left(-\frac{\hbar^2}{2}\Delta^{\mathcal{S}}+U\right)\Psi,$$
    which is the usual and well known expression for the time-dependent Schrödinger equation (except for the superindex $\mathcal{S}$ which is just a notation issue).

\bigskip

  \section{Waves in Classical Mechanics}

\subsection{Time. Action. Length}

Before having a metric on $M$, we can not talk about time of travel of a (no-parameterized) curve in $T^*M$.

The \emph{action} is a function defined on the set of all curves $\gamma$ in $T^*M$, namely, the integral $\int_\gamma\theta$.

Once a metric $T_2$ on $M$ is given, the identification of $TM$ with $T^*M$ produced by that metric allows us to transport from one to the other fiber bundle the notions of time and action. The Liouville form $\theta$ transported to $TM$ gives a representative $\theta/\dot\theta$ of the class of time on the open set out of the quadric of light $\dot\theta=0$. In the case of $T_2$ being positive definite, out of $0$ section. The form $\theta/\dot\theta$ allows us to associate  a ``time'' to each curve in $T^*M$ or $TM$, so generalizing the given one for trajectories of second order differential equations.

Once given the metric $T_2$ on $M$, we define on $TM$ the \emph{length element} $\lambda:=\theta/\|\theta\|=\theta/\sqrt{|\dot\theta|}$.

When $T_2$ is positive definite, $\lambda$ is defined on the whole $TM$ except on the 0 section. In general, $\lambda$ is defined out of the ``quadric of light'' $\dot\theta=0$.

The form $\lambda$ is invariant under the group of homotheties on fibres of $TM$, whose infinitesimal generator is $\dot x^j\partial/\partial \dot x^j$.

The manifold of orbits of this group is the space $J_1^1M$ of 1-jets of curves of $M$. It follows that $\lambda$ is projectable onto $J_1^1M$ (see \cite{Saunders,Primary spectrum} and references therein). Each (no parameterized) curve $\gamma$ in $M$ canonically gives us a curve $\widetilde\gamma$ in $J_1^1M$. The integral $\int_{\widetilde\gamma}\lambda$ is the \emph{length of $\gamma$}.

It does not make sense to talk about time of travel or action along no parameterized curves in $M$, but talk about length.

Length is the ``proper time'' in Relativity.

Parametrization of the trajectories of a classical mechanical system is carried out by the time (the class of time). When, in addition, the trajectories of the system are also parameterizable by the proper time, out of the quadric of light, we have $\theta/\dot\theta=\text{const.}\theta/\sqrt{|\dot\theta|}$ on each trajectory (const. depends on the trajectory). That means $\dot\theta$ is constant along each trajectory. This is the distinctive characteristic of ``relativistic'' systems.
A mechanical system $(M,T_2,\alpha)$ is relativistic if and only if the work form $\alpha$ belongs to the contact system \cite{RelatividadMunoz,RM}. It follows  that there are no systems which are simultaneously relativistic and conservative, except the geodesic one, $\alpha=0$.

Condition $\dot\theta=\text{const.}$ on each trajectory also characterizes those systems whose trajectories are parameterizable by the action. For our subject of study, which are the conservative systems, the only ones satisfying the mentioned condition are the geodesic ones.

The fact that length of a curve in $M$ is not depending on its parametrization allows us to translate the Maupertuis Principle  to the language of geodesics, as it is well known (although often explained with little clarity): let us consider a conservative mechanical system $(M,T_2,dU)$, with Hamiltonian $H$. For a given energy level  $H=E$, let us consider the metric $T_{2,E}:=2(E-U)T_2$. The translation of the Liouville from $\theta$ form $T^*M$ to $TM$ by means of that metric is $\theta_E:=2(E-U)\theta$, where $\theta$ denotes the translation when we use $T_2$ instead.  The length element for $T_{2,E}$ is
$$\lambda_E=\frac{\theta_E}{\sqrt{|\dot\theta_E|}}=\frac{2(E-U)\,\,\theta}{\sqrt{2(E-U)}\sqrt{|\dot\theta|}}=\sqrt{\frac{2(E-U)}{2T}}\,\theta;$$
as, on the hypersurface $H=E$ of $TM$ we have $E-U=T$, the specialization of the 1-forms $\lambda_E$, $\theta$, to $H=E$ coincide. By parameterizing each given curve in $M$ in such a way that its lifting to $TM$ is included into $H=E$, the Maupertuis Principle derived from (\ref{Newtonconservativo}) is translated in this way: for each energy level $H=E$, the trajectories of the system in $M$ are the geodesic of the metric $T_{2,E}=2(E-U)T_2$, and the parametrization of the trajectories (by the ``time'' in $TM$) is the one making the lifting to each of these trajectories belong in the hypersurface $H=E$ of $TM$. (A computation gives $$D=D_{G,E}-\frac{\dot U}{E-U}V$$
on the hypersurface $H=E$. $D_{G,E}$ denotes the geodesic field for the metric $T_{2,E}$, $V$ the infinitesimal generator  of the group of homotheties in fibres, and $D$ is the field which governs  the evolution of the system $(M,T_2,dU)$. By observing through the modified metric, $T_{2,E}$, it turns out that the motion of the system is ``rectilinear'', with a force applied in the direction of the motion given by the second term of $D$).
Coming back to the beginning of this section: on each trajectory (in $TM$) of a mechanical system we define, once fixed an initial point, the functions
``action'' $S=\int\theta$; ``time'' $t=\int\theta/\dot\theta$; length $\ell=\int\theta/\sqrt{|\dot\theta|}$; On the given curve, we have
$$\frac{dS}{dt}=2T,\quad \frac{dS}{d\ell}=\sqrt{2T},\quad \frac{d\ell}{dt}=\sqrt{2T},\quad \frac{d\ell}{dS}=\frac 1{\sqrt{2T}}.$$

\subsection{Wave fronts on the Hamilton-Jacobi solutions}

Let us consider the configuration space $M$, with metric $T_2$ of arbitrary signature. The arguments we will use are local in character, without the need of remembering it  every time. In order not to obscure the development of ideas with details of rigor more or less obvious, we will limit ourselves to consider states in $TM$ out of the ``quadric of light'' $\dot\theta=0$. In the most classical case, with $T_2$ positive definite,  the states of equilibrium will be left out of our considerations.

Let $X$ be an $r$-dimensional submanifold of $M$. The set $\widetilde X\subset T^*M$ which consist of all the 1-forms $\alpha_x\in T_x^*M$, with $x\in X$, annihilating $T_xX$, will be called \emph{conormal bundle} of $X$ in $M$. It is a submanifold of $T^*M$ of dimension $n=\dim M$, whatever the number $r$ is. From the very definition of the Liouville form $\theta$ it follows that the specialization of $\theta$ to $\widetilde X$ is 0.

In the isomorphism $TM\simeq T^*M$ determined by the metric,  the manifold corresponding to $\widetilde X$ in $TM$ is the \emph{normal bundle of $X$ in $M$}, which consists of vectors $v_x\in T_xM$, $x\in X$, orthogonal to $T_xX\subset T_xM$.

Conormal bundles $\widetilde X$ are a particular instance of Lagrangian manifolds, which were considered by Lie in his works on contact transformations and first order partial differential equations \cite{Lie1,Lie2}. For this reason, we will call the submanifolds $Z$ of $T^*M$ (or their translations to $TM$) where $\theta$ specializes as $0$, \emph{Lie manifolds}. Such  submanifolds must be of dimension lower or equal than $n$. Those of dimension $n$ are Lagrangian. If $Z$ is a Lie submanifold of dimension $n$, in the open set where the rank of the projection $Z\to M$ reach the maximum $r$, the image of $Z$ is an $r$-dimensional manifold $X$ and $Z=\widetilde X$ (in the open set where $Z$ projects on $X$).

If $Z$ is a Lie manifold, the specializations to $Z$ of the forms $\theta$ (action), $\theta/\dot\theta$ (time), $\theta/\|\theta\|$ (length) are all of them 0: every path in $Z$ joining two of its points $\alpha$, $\beta$, has action 0, time 0 and length 0. For this reason, Lie manifolds are the natural generalization of the notion of point; in fact, the fibre into $T^*M$ over a point in $M$ is the most obvious example of Lie manifold. None of the second order differential equations is tangent to $Z$ (because, such an equation $D$  holds $\langle \theta,D\rangle=\dot\theta=2T\ne 0$). Thus, to join $\alpha$ and $\beta$ with a ``mechanically possible'' trajectory, that is to say, the lifting to $TM$ of a parameterized trajectory of $M$, it is necessary to get out of $Z$. Lie said that the points of $Z$ are united (\emph{Vereignete}). Perhaps it will be appropriate to call them \emph{entangled} because into the set $Z$ there is no distance or time.

Let us consider the mechanical system $(M,T_2,dU)$, with Hamiltonian $H=T+U$. Let $D$ be the corresponding Hamiltonian vector field: $D\lrcorner\,\omega_2+dH=0$.

On each hypersurface $H=E$ of $T^*M$, the radical of the specialization of $\omega_2$ is $D$ (and its multiples). Without need of applying the Jacobi-Lie theory remembered in \ref{subsectionconservatives}, by taking local coordinates in $H=E$ with which $D$ reduces to canonical form $\partial/\partial z$, it is checked that the expression $\omega_2|_{H=E}$ does not contain $dz$ and its coefficients are independent of $z$.

Let $Y$ be a Lie manifold of dimension $n-1$, contained into the hypersurface $H=E$. The second order differential equation $D$ is no-tangent to $Y$ at any point, from which it follows that $n-1$ of the first integrals of $D$ are independent in $Y$. Local equations of $Y$ are, then, of the form $H=E$, $\beta^i=f^i(\alpha^1,\dots,\alpha^{n-1})$, ($i=1,\dots,n-1$), $z=g(\alpha^1,\dots,\alpha^{n-1})$, where the $\alpha$'s and the $\beta$'s are first integrals of $D$ independent of $H$. The manifold $Z$ which results from suppressing the last equation (``letting $z$ free''), has dimension $n$, is solution of $H=E$, and the specialization of $\omega_2$ to it vanishes, because the specialization of $\omega_2$ to $Y$ is null, and $\omega_2$ does not contain $dz$ or $z$ in its expression in $H=E$.

The manifold $Z$ is obtained by ``sliding'' $Y$ along the integral curves of the field $D$. $Z$ is a Lagrangian manifold where $H=E$. It is, therefore, a solution of the Hamilton-Jacobi equation $H(dS)=E$. Here, $S$ is the function (determined up to a additive constant) such that $\theta|_Z=dS$. Local equations of $Z$ are, then, the $p_j=\partial S/\partial x^j$ ($j=1,\dots,n$), if $\theta=p_jdx^j$.

     By taking an arbitrary submanifold $X$ of $M$, its conormal bundle $\widetilde X$, and the Lie manifold $Y:=\widetilde X\cap\{H=E\}$, the before explained method  (``method of Cauchy characteristics'', as interpreted by Lie) gives us solutions of the Hamilton-Jacobi equation when the differential equations of motion  are previously solved, that is to say, if previously we know how to reduce $D$ to its canonical form $\partial/\partial z$. The ``Hamilton-Jacobi method'' to solve the motion differential equations starts, inversely, from the knowledge of a complete integral of the Hamilton-Jacobi equation in order to integrate the system of differential equations defined by $D$.

     In the Lagrangian manifold $Z$ obtained from the manifold of initial conditions $X$ and energy level $H=E$, the Liouville form $\theta$ specializes as $dS$. Thus, each hypersurface $S=\textrm{const.}$ of $Z$ is a Lie manifold. These manifolds,  when projected on $M$, will be called \emph{wave fronts} associated with the manifold $Z$, that is to say, with the field $u:=\textrm{grad}\,S$, intermediate integral of the equations of motion. The field $u$ is orthogonal at each point to the front wave passing through that point.

     A point $x_0\in M$ can be taken as manifold of initial conditions. In such a case, $\widetilde x_0=T^*_{x_0}M$ (or $T_{x_0}M$, seen in $TM$) and, for the energy level $H=E$, the manifold of initial conditions $\widetilde x_0\cap\{H=E\}$ is the quadric $T(v_0)+U(x_0)=E$. For each given $v_0$ we get a trajectory of $D$ with initial condition $v_0$. The union of all these trajectories fills an $n$-dimensional manifold $Z$ in $TM$. The projection of $Z$ onto $M$ is a neighborhood of $x_0$, where  the action function is defined, which is usually written $S(x_0,x)$, referring to the initial condition $x_0$. $S(x_0,x)$ is the integral $\int_\gamma\theta$ here $\gamma$ is the solution of the second order differential equation $D$ with initial condition $(x_0,v_0)$. In this way, a function can be defined as ``action'', in a neighborhood of the diagonal in $M\times M$, for each energy level $E$.

     If $X\subset M$ is any manifold of initial conditions, given $x_0\in X$ and $v_0\perp T_{x_0}X$ in the energy level $H=E$, the Lagrangian manifold $Z$ built from $\widetilde X\cap\{H=E\}$ has the trajectory of $D$ with initial condition $(x_0,v_0)$ in common with the one built one from $\widetilde x_0\cap\{H=E\}$. Along this trajectory, the action function $S$ defined on $Z$, coincides with the function $S(x_0,x)$: both of them equal $\int_\gamma\theta$, $\gamma$ being the trajectory of $D$. In addition, the wave fronts at each point of this trajectory are tangent, because both are orthogonal to the projection of $D$ in $M$ at the given point. This is a form of the ``Huygens Principle'': each wave front of the ``action'' associated with a solution of Hamilton-Jacobi equation is the enveloping  of the wave fronts coming from the point sources at the given manifold of initial conditions, and at the same level of energy.

\subsection{Particles time - Waves time}\label{particleswavestime}

Like in the previous section, the considerations will be local in character. Let us fix a manifold of initial conditions $X\subset M$, which determines the Lie manifold $\widetilde X$ in $T^*M$ (or seen in $TM$, keeping notation). Each value of the Energy $H=E$ fixes an initial wave front $X_E:=\widetilde X\cap\{H=E\}$, whose propagation along the trajectories of $D$, which rules the evolution of the system $(M,T_2,dU)$, is a Lagrangian manifold $Z_E$, in which the Liouville form $\theta$ is an exact differential: $\theta|_{Z_E}=dS$. The function $S$ is a solution of the Hamilton-Jacobi equation $H(dS)=E$ (or $H(\textrm{grad}S)=E$ in $TM$). The field $u=\textrm{grad}S$ in $M$ is the ``field of velocities'' of the virtual  particles that move in $M$ obeying the second order differential equation  $D$. At each parameterized trajectory of $u$ we have $u=d/dt$. The lifting of this parameterized trajectory to $TM$ is a trajectory of $D$, contained in the Lagrangian manifold $Z_E$. In $Z_E$ there is a well defined \emph{time} function, that parameterizes  each trajectory of $D$ starting from the manifold $X_E$ of initial conditions: it is
$$t=\int_\gamma\frac\theta{\dot\theta}=\int_\gamma \frac{dS}{2(E-U)},$$
on each trajectory $\gamma$ departing from an initial point in $X_E$.

 For each Lagrangian manifold $Z_E$, the wave fronts $S=\textrm{const.}$, when projected onto $M$, are orthogonal to the field $\textrm{grad}S$, that is the field of velocities of the virtual particles that move according to the equations of evolution of the mechanical system.

 We define in $Z_E$ (and so, also in $M$) the function $\tau:=S/E$. By taking $\tau$ as the ``time'', the wave fronts in $Z_E$ (projected to $M$) move at a steady pace. The velocity of passage from a wave front to another according to the orthogonal trajectories, measured with the time $\tau$ is
 $$v=\frac{d\ell}{d\tau}=\frac{\theta/\|\theta\|}{\theta/E}=\frac E{\|\theta\|}=\frac E{\sqrt{\dot\theta}}=\frac E{\sqrt{2(E-U)}},$$
 while the speed of the displacement of the virtual particles, measured with the time given by the ``class of time'' for all the second order differential equations is
 $$\|u\|=\frac{d\ell}{dt}=\frac{\theta/\|\theta\|}{\theta/\dot\theta}=\frac {\dot\theta}{\|\theta\|}=\sqrt{2(E-U)}.$$

 In \cite{QuantizacionMunozAlonso}, \S 6, we note a  misunderstanding between $t$ and $\tau$ found in \cite{SchrII}.

 A known De Broglie  formula \cite{de Broglie} is, in our notation,
 $$\|u\|=\frac{\partial E}{\partial (E/v)},\quad\text{$U$ in a fixed value}.$$

 This formula offers the following interpretation: $v$ is the velocity of propagation of the phase of the train of waves. $\|u\|$, the velocity  of the particles, is the ``velocity of the group'' of the train of waves.  What happens is that these speeds are measured with different times.

 When the potential $U$ vanishes (geodesic system), the Hamiltonian $H$ is $T$ and on each solution of the Hamilton-Jacobi equation we have $T=(1/2)\dot\theta=E$. For the trajectories contained in this solution, the forms of action $\theta$, length $\theta/\|\theta\|=\theta/\sqrt{2E}$ and time $\theta/\dot\theta=\theta/2E$, are proportional and the relationship between particles time $t$ and waves time $\tau$ is $d\tau/dt=2$. In the absence of potential, there is no difference between the time of particles and the time of waves, except for the unit of measure.

\subsection{The equation ${\partial S}/{\partial E}=t$}

This is the equation (7) in \cite{SommerfeldMecanica}, \S44, which is important in the relationship between ``particles time'', that will be $\partial S/\partial E$, and ``wave time'', that equals $S/E$, on each trajectory of the system. Let us give a proof of that formula adapted to the  language  we have been using. The arguments that follow are local in character. The validity of the formula for any trajectory is derived from its local fulfillment, joining pieces.

Let $(M,T_2,dU)$ be a conservative system with Hamiltonian $H=T+U$ and Hamiltonian field $D$ that governs the evolution.

For each  point $x_0\in M$, the conormal bundle $\widetilde x_0$ equals $T_{x_0}^*M$ (or $T_{x_0}M$ translated to $TM$). Given the energy level $H=E$, the manifold $Y_E=\widetilde x_0\cap\{H=E\}$ is the quadric of $T_{x_0}M$ of equation $T(v_{x_0})+U(x_0)=E$, $v_{x_0}\in T_{x_0}M$. For each $v_{x_0}\in Y_E$, there exists an interval $I$ centered at the origin of $\R$ and a curve $\widetilde\gamma\colon I\to TM$ solution of the vector field $D$ with initial condition $\widetilde\gamma(0)=v_{x_0}$. The projection of $\widetilde\gamma$ to $M$ is a curve $x=\gamma(t)$, solution of second order differential equation $D$ with initial conditions $\gamma(0)=x_0$, $\gamma'(0)=v_{x_0}$. For each $x=\gamma(t)$ in the curve, let us put $$S(E;x_0,x):=\int_{\gamma, 0}^t\theta.$$
All Lagrangian submanifolds contained into the hypersurface $H=E$ are generated by trajectories of $D$, so that those passing through $v_{x_0}$ have in common the curve $\widetilde\gamma$.

On each one of these manifolds, $Z_E$, there is a function ``action'', primitive of $\theta|_{Z_E}$, whose value at $x=\gamma(t)$ is $S(E;x_0,x)$ if we fix the value of the primitive at $0$ for $v_{x_0}$.

As it is well known (and remembered above) the Maupertuis Principle allows us to prove that the trajectories of the field $D$ in the energy level $H=E$ are projected to $M$ as geodesic paths for the metric $2(E-U)T_2$ (in general, the proper parametrization of these geodesic curves does not match with that of trajectories of $D$). The properties of the exponential map $T_{x_0}M\to M$ show that the trajectories of $D$ starting at $Y_E$ are projected to $M$ by filling a neighborhood of $x_0$.

Let us consider on $TM$ (or $T^*M$) the field $W=V/\dot\theta$, where $V$ is the infinitesimal generator of the group of homotheties: $V=\dot x^j\partial/\partial \dot x^j=p_j\partial/\partial p_j$. The field $W$ is vertical and holds $WH=1$. For that, $W$ leaves  the subring of $\C^\infty(TM)$ generated by $\A$ and $H$ stable. In local coordinates, for that subring $W=\partial/\partial H$.

The trajectories of the field $W$ in $TM$ are the same as those of  $V$ (lines passing through 0 on each $T_xM$), but are parameterized by the energy $H$. The one-parametric group generated by $W$ changes, in $TM$, the hypersurfaces $H=\textrm{const.}$ into each other. The field $W$ is not Hamiltonian, so that the group it generates does not transform, in general, Lagrangian submanifolds into Lagrangian submanifolds. On the Liouville form $\theta$ we have $L_W\theta=W\lrcorner\,\omega_2=\theta/\dot\theta$, which belongs to the class of time ($L_W$ denotes the Lie derivative along the vector field $W$).  For each $\widetilde\gamma$ in $TM$, let us denote by $\widetilde\gamma_\epsilon$ the transformed of $\widetilde\gamma$ by the one-parametric group generated by $W$ when the parameter equals $\epsilon$. We have,
$$\left[\frac d{d\epsilon}\int_{\widetilde\gamma_\epsilon}\theta\right]_{\epsilon =0}=\int_{\widetilde\gamma}L_W\theta=\int_{\widetilde\gamma}\frac\theta{\dot\theta}.$$
Now, if $\widetilde\gamma$ is a solution of $D$ considered above, we get
$$\lim_{\epsilon\to 0}\frac{\int_{\widetilde\gamma_\epsilon}\theta-S(E;x_0,x)}\epsilon=\int_{\widetilde\gamma}\frac\theta{\dot\theta}=t,$$
time of travel along $\gamma$ between $x_0$ and $x$.
$\widetilde\gamma_\epsilon$ is a curve into the energy level $H=E+\epsilon$ that projects itself in $M$ as a curve joining $x_0$ with $x$, and it is parameterized by the same  parameter as the curve $\gamma$.

Let $\widetilde\gamma'_\epsilon$ be the curve solution of $D$ in the energy level $H=E+\epsilon$ joining $x_0$ and $x$ and parameterized as solution of $D$, that is to say, by the class of time. The final time for $\widetilde\gamma'_\epsilon$ will be $t'$. By using a correspondence between $[0,t]$ and $[0,t']$ (for example, by means of an homothety), we establish a univocal correspondence between  the points of $\widetilde\gamma$, $\widetilde\gamma_\epsilon$ and $\widetilde\gamma_\epsilon'$. In such correspondence the difference of coordinates between corresponding points, is of order at most $\epsilon$.
Since $(D)$ is the radical of $d\theta$ on the hypersurface $H=E+\epsilon$, the difference $\int_{\widetilde\gamma_\epsilon}\theta-\int_{\widetilde\gamma'_\epsilon}\theta$ is of order at least $\epsilon^2$. For this reason we substitute $\int_{\widetilde\gamma_\epsilon}\theta$ by $\int_{\widetilde\gamma'_\epsilon}\theta=S(E+\epsilon;x_0,x)$, and it gives finally
$$\frac{\partial S(E;x_0,x)}{\partial E}=t,$$
as stated.

\bigskip

\section{Time dependent Schrödinger equation from a Hertz constraint}\label{SchrodingerHertz}

Let us consider a conservative mechanical system $(M,T_2,dU)$, with Hamiltonian $H=T+U$.

By adopting the point of view of Hertz, let us consider this system as the projection of a geodesic system $(\M,\T_2,0)$, like in Section \ref{Hertz}. Let us keep that notation. The metric $\T_2=g_{00}(dx^0)^2+g_{\mu\nu}dx^\mu dx^\nu$ gives a kinetic energy $\T=(1/2)g^{00}p_0^2+(1/2)g^{\mu\nu}p_\mu p_\nu$, with $g^{00}$ independent of $x^0$, which gives the first integral of the equations of motion $p_0=P_0=\textrm{const.}$
  The restriction of $\T$ to the manifold $p_0=P_0$ is $(1/2)g^{00}P_0^2+T$, which coincides with $H$ if $U=(1/2)g^{00}P_0^2$. In the projection $\M\to M$ (in which $\A$ is the ring of first integrals of the field $\textrm{grad}\,x^0$ in $\M$), the term $(1/2)g^{00}p_0^2$ of the kinetic energy in $T\M$ is ``transferred'' from the hypersurface $p_0=P_0$ to $M$ as potential energy. This transference is possible because $g^{00}$ does not depend on $x^0$.

  The Hamilton-Jacobi equation for the Hamiltonian $\H=\T$  is
  \begin{equation}\label{HJHertz}
  g^{00}\left(\frac{\partial \SS}{\partial x^0}\right)^2+g^{\mu\nu}\frac{\partial \SS}{\partial x^\mu}\frac{\partial \SS}{\partial x^\nu}=2E.
  \end{equation}

  Functions $\H$, $p_0$ are in involution with respect to the symplectic structure in $T^*\M$, so we can apply the method of Jacobi and looking for common solutions with (\ref{HJHertz}) and the Hamilton-Jacobi  equation  corresponding to the ``Hamiltonian'' $p_0$. That equation is $\partial\SS/\partial x^0=P_0$, for each constant $P_0$. Let us put, then,
  \begin{equation}\label{HJHertz2}
  \SS=P_0x^0+S(x^1,\dots,x^n),
  \end{equation}
  with what (\ref{HJHertz}) is
  \begin{equation}\label{HJHertz3}
  g^{00}P_0^2+g^{\mu\nu}\frac{\partial S}{\partial x^\mu}\frac{\partial S}{\partial x^\nu}=2E,
  \end{equation}
  that is, the Hamilton-Jacobi equation for $(M,T_2,dU)$, when $U=(1/2)g^{00}P_0^2$, as above.

  Given the mechanical system $(M,T_2,dU)$ and a solution $S$ of the Hamilton-Jacobi equation:
  \begin{equation}\label{HJHertz4}
  g^{\mu\nu}\frac{\partial S}{\partial x^\mu}\frac{\partial S}{\partial x^\nu}+2U=2E,
  \end{equation}
  we come back, by defining $\M:=M\times\R$, with metric $\T_2$ as above, taking $P_0=-E/c$ (here, $c$ is a universal constant) and $g_{00}$ defined by $U=(1/2)g^{00}P_0^2$.
  Then, each solution of (\ref{HJHertz4}) gives us a solution  (\ref{HJHertz2}) of (\ref{HJHertz}). Conversely, each solution of (\ref{HJHertz}) of the form (\ref{HJHertz2}) gives a solution of (\ref{HJHertz4}).

  The construction of the action $\SS$ starting from a manifold of initial conditions in $\M$ goes as follows: we fix as manifold of initial conditions in $\M$ a  submanifold $M_0$ parameterized by $M$. In local coordinates, $M_0$ has the equations \begin{equation}\label{HJHertz41}
  x^0=\xi(x^{1},\dots,x^n).
  \end{equation}
  The 1-forms of $T^*\M$ that specialize as 0 in $M_0$ are the multiples of the 1-form  $dx^0-(\partial\xi/\partial x^{\mu})dx^{\mu}$, so that the equations of the conormal bundle $\widetilde M_0$ in $T^*\M$ are
  (\ref{HJHertz41}) joint with
  \begin{equation}\label{HJHertz51}
  p_{\mu}=-p_0\frac{\partial\xi}{\partial x^{\mu}}, 
  \end{equation}
  and, as coordinates on $\widetilde M_0$ serve $p_0,x^{1},\dots,x^n$.

  The submanifold of $\widetilde M_0$ corresponding to the energy level $\T=E$  has as equations (\ref{HJHertz41}), (\ref{HJHertz51}) and
  \begin{equation}\label{HJHertz52}
  g^{00}p_0^2+g^{\mu\nu}p_\mu p_\nu=2E. 
  \end{equation}
  By slicing with the hipersurface $p_0=P_0$ of the Hertz constraint  it remains the submanifold of $\widetilde M_0$, of dimension $n-1$, with local equations (\ref{HJHertz41}), (\ref{HJHertz51}), (\ref{HJHertz52}) and
  \begin{equation}\label{HJHertz53}
  p_0=P_0. 
  \end{equation}

  This submanifold $\Y_{n-1}$ is propagated along the trajectories of the (commuting) fields $\DD$ (the Hamiltonian field in $T\M$ corresponding to $\H=\T$), $\partial/\partial x^0$ (the Hamiltonian field in $T\M$ corresponding to $p_0$) and generates a Lagrangian manifold $\Z_{n+1}$, parametrizable by $x^0,x^1,\dots,x^n$, where it holds $\theta|_{\Z}=d\SS$ for a function $\SS=\SS(x^0,x^1,\dots,x^n)$. From its very construction, $\Z$ is a  solution of the Hamilton-Jacobi equation (\ref{HJHertz}) and also of the equation $\partial\SS/\partial x^0=P_0$ (as it is  derived from condition (\ref{HJHertz53}) for $\Y$). In particular $\SS=P_0x^0+S(x^1,\dots,x^n)$ for some function $S$. Thus, the equations of $\Z$ are:
  \begin{equation}\label{HJHertz6}
  p_0=\frac{\partial\SS}{\partial x^0}=P_0,\quad p_\mu=\frac{\partial \SS}{\partial x^\mu}=\frac{\partial S}{\partial x^\mu},\quad (\mu=1,\dots,n).
  \end{equation}

  The projection $T\M\to TM$ (from which it is derived $T^*\M\to T^*M$, via the metric) sends $\Y_{n-1}$ to the submanifold $Y_{n-1}$ obtained from the equations (\ref{HJHertz41}), (\ref{HJHertz51}), (\ref{HJHertz52}), (\ref{HJHertz53}) putting aside $x^0$ and considering the last of them as a constraint, that allows us to go without $p_0$. The equations of $Y_{n-1}$ are
  \begin{align}
  & p_\mu=-P_0\frac{\partial\xi}{\partial x^\mu}\quad (\mu=1,\dots,n)\label{HJHertz7}\\[0.5em]
  & g^{00}P_0^2+g^{\mu\nu}p_\mu p_\nu=2E\label{HJHertz7'}
  \end{align}

  Expression (\ref{HJHertz7}), when substituted into (\ref{HJHertz7'}), results on the equations of an hypersurface $X_{n-1}\subset M$. So that (\ref{HJHertz7'}) selects in $\widetilde X_{n-1}$ the energy level $H=E$, with Hamiltonian $H=T+U$, $U=(1/2)g^{00}P_0^2$.

  Manifold $\Z_{n+1}$ is projected into $TM$ as the submanifold $Z_n$, whose equations are the second group of (\ref{HJHertz6}). $Z_n$ is obtained by sliding $Y_{n-1}$ along the trajectories of the Hamiltonian field $D$ in $TM$, projection of $\DD$ in $T\M$.

  The time of the waves $\mathbb{t}$ in $\Z_{n+1}$ is
  \begin{equation}\label{HJHErtz8}
  \t=\frac{\SS}{E}=\frac{P_0}E x^0+\frac SE=-\frac{x^0}c+\tau
  \end{equation}
  where $\tau$ is the time of waves in $Z_n$. As $(\M,\T_2,0)$ is a system free of forces, the time $\t$ is (except for the factor 2) the time of particles in $\Z_{n+1}$, which is also the time of particles (absolute time) in $Z_n$. The time $x^0/c$ comes from the ``clock'' that we put when replacing the configuration space $M$ with $\M=M\times \R$.

  The quantization in the system $(\M,\T_2,0)$ gives
  \begin{equation}\label{HJHertz9}
  \widehat p_0=-i\hbar\frac{\partial}{\partial x^0},\quad \widehat{\T}=-\frac 12\hbar^2\Ddelta,
  \end{equation}
  where $\Ddelta$ is the Laplacian operator associated with $\T_2$ in $\M$. We have
   \begin{equation}\label{HJHertz10}
  \Ddelta=g^{00}\widehat{p_0^2}+\Delta,
  \end{equation}
 where $\Delta$ is the Laplacian operator in $(M,T_2)$ lifted to $\M$ (the tensor $T_2$ lifted to $\M$ and quantized there).

 According to the formulae given in Section \ref{tres}, it holds
 \begin{equation}\label{HJHertz11}
  \widehat{p_0^2}=-\hbar^2\left(\frac{\partial ^2}{(\partial x^0)^2}-\Gamma_{00}^\mu\frac{\partial}{\partial x^\mu}\right)=
    -\hbar^2\left(\frac{\partial ^2}{(\partial x^0)^2}-\frac 12\textrm{grad}(g_{00})\right).
  \end{equation}

  In general, $ \widehat{p_0^2}$ does not commute with $\Ddelta$, which prevents us to separate variables (the $x^0$ from the $x^\mu$) in the Schrödinger equation for $(\T,\M)$.

  The classical magnitude $H=T+U$, lifted from $TM$ to $T\M$ coincides with the classical magnitude $\T$ in the hypersurface $p_0=P_0$. The quantization of $H$ in $M$, and $\M$, has the form
  \begin{equation}\label{HJHertz12}
  \widehat{H}=\widehat T+U=-\frac 12\hbar^2\Delta+U.
  \end{equation}
  The operators  $\widehat{p_0}$, $\widehat H$, commute in $\C^\infty(\M)$, and we can apply the method of separation of variables in order to  find the functions that simultaneously satisfy the corresponding wave equations in the classical state corresponding to the action $\SS$. In such an state of the system $(\M,\T_2,0)$ is $p_0=P_0$, $H=E$($=\T$), and the wave equations for $p_0$, $H$, are:
  \begin{equation}\label{HJHertz13}
  -i\hbar\frac{\partial \Psi}{\partial x^0}=P_0\Psi,\qquad \left(-\frac 12\hbar^2\Delta+U\right)\Psi=E\Psi
  \end{equation}

  The first equation gives
  $$\Psi=e^{i\frac{P_0}{\hbar}x^0}\Phi(x^1,\dots,x^n)$$
  and then, the second one, becomes
  $$\left(-\frac 12\hbar^2\Delta+U\right)\Phi=E\Phi.$$
  The  Schrödinger equation with time $x^0/c=t^0$
  \begin{equation}\label{HJHertz14}
  \widehat H\Psi=i\hbar\,\frac{\partial\Psi}{\partial t^0}
  \end{equation}
  holds if $P_0=-E/c$. The time $t^0$ in these equations is $\tau-2t$, $t$ being the absolute time and $\tau$ the waves time on the given solution of the Hamilton-Jacobi equation.
 \bigskip

  Due to its linear character, (\ref{HJHertz14}) is valid for any superposition of states which hold it. Finally, such equation can be interpreted, by using the Stone theorem \cite{Yosida}: $-(i/\hbar)\hat H$ generates a uniparametric group of unitary automorphisms of the Hilbert space. $t^0$  is the parameter of the group. However, $t^0$ \emph{is not the time-duration} of the Classical Mechanics. The interpretation of $t^0$ in classical terms holds at each  stationary state associated with a solution of the Hamilton-Jacobi equation where $t^0=\tau-2t$ ($\tau=S/E$). But we do not have an interpretation of the parameter $t^0$ for a superposition of such states.
  \bigskip

  \section{On the interpretation of $E=h\nu$}\label{frecuencia}

  When the Schrödinger equation is derived from a time constraint, the $t$ can be interpreted as the time duration that is imposed on the quantum system by a classical mechanical system into which is submersed. For example, in the non relativistic theory of radiation \cite{Neumann}.

 Nevertheless, when a quantum system is considered by itself, as it is done in Section \ref{SchrodingerHertz}, the $t^0$-parameter does not admit an interpretation in a classical sense (although it is interpretable inside each solution of the stationary Hamilton-Jacobi equation). Then, the issue of whether frequencies $\nu^0=E/h$ are interpretable in classical terms in a way concordant with the experience arises. It must happen that, for the systems whose classical solutions are periodic, $\nu^0$ coincides with the classical period.

 In the case of  particles that move under a central force, there are only two types of forces that, necessarily, produce closed periodic trajectories (below a certain energy level): the elastic force and the Kepler force (Bertrand theorem \cite{Tisserand}). 
 
 For the harmonic oscillator, the periods of $t$, $\tau$ and $t^0$ coincide.
 
 For the Kepler problem, the relationship between the period $T$ of a revolution and the energy $E$ is $T^2=-K^2E^{-3}$ ($K$ is a constant given in the considered system). The formula $\partial S/\partial E=T$ for the total action corresponding to a cycle gives $S=2K(-E)^{-1/2}$ (taking $S=0$ for $E=-\infty$) and, hence, $S/E=-2T$, so that $\tau-2t$ in a complete cycle is $-4T$. 
 
 In both cases, if $\nu$ is interpreted as the classical frequency (with respect to the absolute time), along each cycle the quantum wave function recovers its initial phase.


\end{document}